\documentclass[final,3p]{elsarticle}
\usepackage{amsmath,amssymb,graphicx}
\usepackage{amsfonts,amsthm}

\usepackage{amsmath,amssymb,graphicx,amsthm}
\usepackage{color}
\usepackage{soul}

\newcommand{\D}{{\mathrm{d}}}
\newtheorem{theorem}{Theorem}

\newtheorem{proposition}{Proposition}

\newtheorem{corollary}{Corollary}

\begin{document}

\begin{frontmatter}

\title{Thermodynamics in the Limit of Irreversible Reactions}
\author{A.~N.~Gorban}
 \ead{ag153@le.ac.uk}
\address{Department of Mathematics, University of Leicester, Leicester, LE1 7RH, UK}
\author{E.~M.~Mirkes}
\address{Institute of Space and Information Technologies, Siberian Federal University,
Krasnoyarsk, Russia}
\author{G.~S.~Yablonsky}
\address{Parks College, Department of Chemistry, Saint Louis
University, Saint Louis, MO 63103, USA}



\begin{abstract}
For many real physico-chemical complex systems detailed mechanism includes both
reversible and irreversible reactions. Such systems are typical in homogeneous combustion
and heterogeneous catalytic oxidation. Most complex enzyme reactions include irreversible
steps. The classical thermodynamics has no limit for irreversible reactions whereas the
kinetic equations may have such a limit. We represent the systems with irreversible
reactions as the limits of the fully reversible systems when some of the equilibrium
concentrations tend to zero. The structure of the limit reaction system crucially depends
on the relative rates of this tendency to zero. We study the dynamics of the limit system
and describe its limit behavior as $t \to \infty$. If the reversible systems obey the
principle of detailed balance then the limit system with some irreversible reactions must
satisfy the {\em extended principle of detailed balance}. It is formulated and proven in
the form of two conditions: (i) the reversible part satisfies the principle of detailed
balance and (ii) the convex hull of the stoichiometric vectors of the irreversible
reactions does not intersect the linear span of the stoichiometric vectors of the
reversible reactions. These conditions imply the existence of the global Lyapunov
functionals and alow an algebraic description of the limit behavior. The thermodynamic
theory of the irreversible limit of reversible reactions is illustrated by the analysis
of hydrogen combustion.
\end{abstract}

\end{frontmatter}

\section{Introduction\label{sec1}}

\subsection{The problem: non-existence  of thermodynamic functions in the limit of irreversible reactions}

We consider a homogeneous chemical system with $n$ components $A_i$, the
concentration of $A_i$ is $c_i \geq 0$, the amount of $A_i$ in the system is
$N_i\geq 0$, $V$ is the volume,  $N_i=Vc_i$, $T$ is the temperature. The $n$
dimensional vectors $c=(c_i)$ and $N=(N_i)$ belong to the closed positive
orthant $\mathbb{R}^n_+$ in $\mathbb{R}^n$. ($\mathbb{R}^n_+$ is the set of all
vectors $x\in \mathbb{R}^n$ such that $x_i\geq 0$ for all $i$.)

{\em The classical thermodynamics has no limit for irreversible reactions whereas the
kinetic equations have.} For example, let us consider a simple cycle
$$A_1\underset{k_{-1}}{ \overset{k_{1}}{\rightleftharpoons}}  A_2
\underset{k_{-2}}{ \overset{k_{2}}{\rightleftharpoons}} A_3 \underset{k_{-3}}{
\overset{k_{3}}{\rightleftharpoons}} A_1$$ with the equilibrium concentrations
$c^{\rm eq}=(c_1^{\rm eq},c_2^{\rm eq},c_3^{\rm eq})$ and the detailed balance
conditions:
$$k_i c^{\rm eq}_i=k_{-i} c^{\rm eq}_{i+1}$$
under the standard cyclic convention, here, $A_{3+1}=A_1$ and $c_{3+1}=c_1$.
The perfect free energy has the form
$$F= \sum_i RTV c_i \left(\ln\left(\frac{c_i}{c_i^{\rm eq}}\right)-1\right) +
const\, .$$

Let the equilibrium concentration $c_1^{\rm eq} \to 0$ for the fixed values of
$c_{2,3}^{\rm eq}>0$. This means that
$$\frac{k_{-1}}{k_1}=\frac{c_1^{\rm eq}}{c_2^{\rm eq}}\to 0 \mbox{
 and  } \frac{k_{3}}{k_{-3}}=\frac{c_1^{\rm eq}}{c_3^{\rm eq}}\to 0 \, .$$
Let us take the fixed values of the rate constants $k_1$, $k_{\pm 2}$ and
$k_{-3}$. Then the limit kinetic system exists and has the form:
$$A_1{ \overset{k_{1}}{\rightarrow}}  A_2
\underset{k_{-2}}{ \overset{k_{2}}{\rightleftharpoons}} A_3
\underset{k_{-3}}{\leftarrow} A_1 \, .$$ It is a routine task to write a first
order kinetic equation for this scheme. At the same time, the free energy
function  $F$ has no limit: it tends to $\infty$ for any positive vector of
concentrations because the term $c_1 \ln({c_1}/{c_1^{\rm eq}})$ increases to
$\infty$. The free energy cannot be normalized by adding a constant term
because the variation of the term $c_1 \ln({c_1}/{c_1^{\rm eq}})$ on an
interval $[0,\overline{c}]$ with fixed $\overline{c}$ also increases to
$\infty$, it varies from $-c^{\rm eq}_1/e$ (for the minimizer, $c_1={c_1^{\rm
eq}}/e$) to a large number $\overline{c}(\ln \overline{c}-\ln {c_1^{\rm eq}}) $
(for $c_1=\overline{c}$).

The logarithmic singularity is rather ``soft" and does not cause a real
physical problem because even for ${c_1^{\rm eq}}/{c_1}=10^{-10}$ the
corresponding large term in the free energy will be just $\sim 23RT$ per mole.
Nevertheless, the absence of the limit causes some mathematical questions. For
example, the density,
\begin{equation}\label{perfectFE}
f=F/(RTV)=\sum_i c_i (\ln (c_i/c_i^{\rm eq})-1) \, ,
\end{equation}
is a Lyapunov function for a system of chemical kinetics for a perfect mixture
with detailed balance under isochoric isothermal conditions. Here, $c_i$ is the
concentration of the $i$th component and $c_i^{\rm eq}$ is its equilibrium
concentration for a selected value of the linear conservation laws, the
so-called ``reference equilibrium".

This function is used for analysis of stability, existence and uniqueness of
chemical equilibrium since the work of Zeldovich (1938, reprinted in 1996
\cite{Zeld}). Detailed analysis of the connections between detailed balance and
the free energy function was provided in \cite{ShapiroShapley1965}. Perhaps,
the first detailed proof that  $f$ is a Lyapunov function for chemical kinetics
of perfect systems with detailed balance was published in 1975
\cite{VolpertKhudyaev1985}. Of course, it does not differ significantly from
the Boltzman's proof of his $H$-theorem (1873 \cite{Boltzmann}).

For the irreversible systems which are obtained as limits of the systems with
detailed balance, we should expect the preservation of stability of the
equilibrium. Moreover, one can expect existence of the Lyapunov functions which
are as universal as the thermodynamic functions are. The ``universality" means
that these functions depend on the list of components and on the equilibrium
concentrations but do not depend on the reaction rate constants directly.

The thermodynamic potential of a component $A_i$ cannot be defined in the irreversible
limits when the equilibrium concentration of $A_i$ tends to 0. Nevertheless, in this
paper, we construct the universal Lyapunov functions for systems with some irreversible
reactions. Instead of detailed balance we use the weaker assumption that these systems
can be obtained from the systems with detailed balance when some constants tend to zero.

\subsection{The extended form of detailed balance conditions for systems with irreversible reactions}

Let us consider a reaction mechanism in the form of the system of
stoichiometric equations
\begin{equation}\label{ReactMech}
\sum_i \alpha_{ri} A_i \to \sum_j \beta_{rj} A_j \;\; (r=1, \ldots, m) \, ,
\end{equation}
where  $\alpha_{ri}\geq 0$, $\beta_{rj}\geq 0$ are the stoichiometric coefficients. The
reverse reactions with positive rate constants are included in the list (\ref{ReactMech})
separately (if they exist). The stoichiometric vector $\gamma_r$ of the elementary
reaction is $\gamma_r=(\gamma_{ri})$, $\gamma_{ri}=\beta_{ri}-\alpha_{ri}$. We always
assume that there exists a strictly positive conservation law, a vector $b=(b_i)$,
$b_i>0$ and $\sum_i b_i \gamma_{ri}=0$ for all $r$. This may be the conservation of mass
or of total number of atoms, for example.

According to the {\em generalized mass action law}, the reaction rate for an elementary
reaction (\ref{ReactMech}) is (compare to Eqs. (4), (7), and (14) in \cite{Grmela2012}
and Eq. (4.10) in \cite{GiovangigliMatus2012})
\begin{equation}\label{GenMAL}
w_r=k_r \prod_{i=1}^n a_i^{\alpha_{ri}} \, ,
\end{equation}
where $a_i\geq 0$ is the {\em activity} of $A_i$,
\begin{equation}\label{StandardActivity}
a_i=\exp\left(\frac{\mu_i-\mu_i^0}{RT}\right)\, .
\end{equation}
Here, $\mu_i$ is the chemical potential and  $\mu_i^0$ is the standard chemical
potential of the component $A_i$.

This law has a long history (see
\cite{Feinberg1972_a,Yablonskiiatal1991,Grmela2010,GiovangigliMatus2012}). It was
invented in order to meet the thermodynamic restrictions on kinetics. For this purposes,
according to the principle of detailed balance, the rate of the reverse reaction is
defined by the same formula and its rate constant should be found from the detailed
balance condition at a given equilibrium.

It is worth mentioning that the free energy has no limit when some of the reaction
equilibrium constants tend to zero. For example, for the ideal gas the chemical potential
is $\mu_i(c,T)=RT\ln c_i+\mu _i^0(T)$. In the irreversible limit some $\mu_i^0 \to
\infty$. On the contrary, the activities remain finite (for the ideal gases $a_i=c_i$)
and the approach based on the generalized mass action law and the detailed balance
equations $w_r^+=w_r^-$ can be applied to find the irreversible limit.

The list (\ref{ReactMech}) includes reactions with the reaction rate constants
$k_r>0$. For each $r$ we define $k_r^+=k_r$, $w_r^+=w_r$, $k_r^-$ is the
reaction rate constant for the reverse reaction if it is on the list
(\ref{ReactMech}) and 0 if it is not, $w_r^-$ is the reaction rate for the
reverse reaction if it is on the list (\ref{ReactMech}) and 0 if it is not. For
a reversible reaction, $K_r=k_r^+/k_r^-$

The principle of detailed balance for the generalized mass action law is: For
given values $k_r$ there exists a positive equilibrium $a_i^{\rm eq}>0$ with
detailed balance, $w_r^+=w_r^-$.

Recently, we have formulated the extended form of the detailed balance
conditions for the systems with some irreversible reactions
\cite{GorbYabCES2012}. This {\em extended principle of detailed balance} is
valid for all systems which obey the generalized mass action law and are the
limits of the systems with detailed balance when some of the reaction rate
constants tend to zero. It consists of two parts:
\begin{itemize}
\item{The {\em algebraic condition}: The principle of detailed balance is valid for the reversible
part. (This means that for the set of all reversible reactions there exists a
positive equilibrium where all the elementary reactions are equilibrated by
their reverse reactions.)}
\item{ The {\em structural condition}: The convex hull of the stoichiometric
vectors of the irreversible reactions has empty intersection with the linear
span of the stoichiometric vectors of the reversible reactions. (Physically,
this means that the irreversible reactions cannot be included in oriented
cyclic pathways.)}
\end{itemize}
Let us recall the formal convention: the linear span of empty set is $\{0\}$,
the convex hull of empty set is empty.

In our previous work \cite{GorbYabCES2012} we studied the systems with some irreversible
reactions which are the limits of the reversible systems with detailed balance. The
structural and algebraic conditions were found. The present paper is focused on the
dynamical consequences of these conditions. We prove that the attractors always consist
of the fixed points. These limit points (``partial equilibria") are situated on the faces
of the positive orthant of concentrations. These faces and the partial equilibria are
described in the paper.

\subsection{The structure of the paper}

In Sec.~\ref{Sec:Multiscale} we study the systems with detailed balance, their {\em
multiscale} limits and the limit systems which satisfy the extended principle of detailed
balance. The classical {\em Wegscheider identities} for the reaction rate constants are
presented. Their limits when some of the equilibria tend to zero give the extended
principle of detailed balance.

We use the {\em generalized mass action law} for the reaction rates. For the analysis of
equilibria for the general systems, the formulas with activities are the same as for the
ideal systems and it is convenient to work with activities unless we need to study
dynamics. The dynamical variables are amounts and concentrations. In a special
subsection~\ref{SecActConc} we discuss the relations between concentration and
activities, formulate the main assumptions and present formulas for the dissipation rate.

We introduce {\em attractors} of the systems with some irreversible reactions and study
them in Sec.~\ref{Sec:Attr}. It includes the central results of the paper. We fully
characterize the faces of the positive orthant that include $\omega$-limit sets. On such
a face, dynamics is completely degenerated (zero rates) or it is driven by a smaller
reversible system that obeys classical thermodynamics.

{\em Hydrogen combustion} is the most studied and very important gas reaction.
It serves as a main benchmark example for many studies of chemical kinetics.
This is already not a toy example but the complexity of this system is not
extremely high: in the usual models there are 6-8 components and $\sim$15-30
elementary reversible reactions. Under various conditions some of these
reactions are practically irreversible. We use this system as a benchmark in
Sec.~\ref{Sec:Hydro} and give an example of the correct separation of the
reactions into reversible and irreversible part. The limit behavior of this
system in time is described.

In Conclusion we briefly discuss the results with focus on the {\em unsolved problems}.

\section{Multiscale limit of a system with detailed balance \label{Sec:Multiscale}}
\subsection{Two classical approaches to the detailed balance condition}

There are two traditional approach to the description of the reversible systems
with detailed balance. First, we can start from the independent rate constants
of the elementary reactions and consider the solvability of the detailed
balance equations as the additional condition on the admissible values of the
rate constants. Here, for $m$ elementary reactions we have $m$ constants ($m$
should be an even number, $m=2\ell$, $\ell=m/2$) and some equations which
describe connections between these constants. This approach was introduced by
Wegscheider in 1901 \cite{Wegscheider1901} and developed further by many
authors \cite{ShuShu1989,Colquhoun2004}.

Secondly, we can select a ``forward" reaction in each pair of mutually reverse
elementary reactions. If a positive equilibrium is known then we can find the
reaction rate constants for the reverse reaction from the constants for forward
reaction and the detailed balance equations. Therefore,  the forward reaction
rate constants and a set of the equilibrium activities form the complete
description of the reaction. Here we have $\ell+n$ independent constants,
$\ell=m/2$ rate constants of forward reactions and $n$ (it is the number of
components) equilibrium activities. For these $\ell+n$ constants, the principle
of detailed balance produces no restrictions. This second approach is popular
in applied chemical thermodynamics and kinetics
\cite{PrigogineDefay1962,GorbanMirFGV1989,YangHlavacek2006} because it is
convenient to work with the independent parameters ``from scratch".

The Wegscheider conditions appear as the necessary and sufficient  conditions
of solvability of the detailed balance equations. (See, for example, the
textbook \cite{Yablonskiiatal1991}). Let us join the forward and reverse
elementary reactions and write
\begin{equation}\label{ReactMechRev}
\sum_i \alpha_{ri} A_i \rightleftharpoons \sum_j \beta_{rj} A_j \;\; (r=1,
\ldots, \ell) \, .
\end{equation}
The {\em stoichiometric matrix} is $\boldsymbol{\Gamma}=(\gamma_{ri})$,
$\gamma_{ri}=\beta_{ri}-\alpha_{ri}$ (gain minus loss). The {\em stoichiometric
vector} $\gamma_r$ is the $r$th row of $\boldsymbol{\Gamma}$ with coordinates
$\gamma_{ri}=\beta_{ri}-\alpha_{ri}$.

Both sides of the detailed balance equations, $w_r^+=w_r^-$, are positive for
positive activities. The solvability of this system for positive activities
means the solvability of the following system of linear equations:
\begin{equation}\label{DetBalGen}
\sum_i \gamma_{ri} x_i = \ln k_r^+-\ln k_r^-=\ln K_r \;\; (r=1, \ldots \ell )
\end{equation}
($x_i=\ln a_i^{\rm eq}$). Of course, we assume that if $k_r^+>0$ then $k_r^->0$
(reversibility) and the equilibrium constant $K_r> 0 $ is defined for all
reactions from (\ref{ReactMechRev}).

\begin{proposition}\label{Prop:Wegscheider}
The necessary and sufficient conditions for existence of the positive
equilibrium $a_i^{\rm eq}>0$ with detailed balance is: For any solution $
\boldsymbol{\lambda}=    (\lambda_r)$ of the system
\begin{equation}\label{lambdaGamma}
\boldsymbol{\lambda \Gamma} =0  \;\; \left(\mbox{i.e.}\;\; \sum_{r=1}^{\ell}
\lambda_r \gamma_{ri}=0\;\; \mbox{for all} \;\; i\right)
\end{equation}
the Wegscheider identity holds:
\begin{equation}\label{WegscheiderLambda}
\prod_{r=1}^{\ell}     (k_r^+)^{\lambda_r}=\prod_{r=1}^{\ell}
(k_r^-)^{\lambda_r} \, .
\end{equation}
It is sufficient to use in (\ref{WegscheiderLambda}) any basis of solutions of the system
(\ref{lambdaGamma}): $\boldsymbol{\lambda} \in \{\boldsymbol{\lambda}^1, \cdots ,
\boldsymbol{\lambda}^{q}\}$.
\end{proposition}

\subsection{Multiscale degeneration of equilibria \label{Sec:Shifted}}

We consider the systems with some irreversible reactions as the limits of the fully
reversible systems when some reaction rate constants tend to zero. In the reversible
systems, the principle of detailed balance implies the Wegscheider identities
(\ref{WegscheiderLambda}). Therefore, the limit system is not arbitrary. Some
consequences of the Wegscheider identities persist though a part of reaction rate
constants in these identities become zero. In \cite{GorbYabCES2012} we compare these
consequences with the grin of the Cheshire cat: the whole cat (the reversible system with
detailed balance) vanishes but the grin persists.

We can postulate that some reaction rate constants go to zero.  However, the reaction
rate constants are not independent. They are connected by the Wegscheider identities. The
rate constants should tend to zero with preservation of their relations. Therefore, the
simple strategy, just to neglect the rates of some of the reactions, cannot be applied
for complex reactions. Nevertheless, we can change the variables and use the independent
set ``reaction rate constants for the forward reactions + equilibrium activities" (see
\cite{PrigogineDefay1962,GorbanMirFGV1989,YangHlavacek2006,GorbYabCES2012}). Every set of
positive values of these variables corresponds to a reversible system with detailed
balance and no additional restrictions are needed. If the reversible system degenerates
to a system with some irreversible reactions then some of the equilibrium activities tend
to zero. Let us study this process of degeneration of reversible reactions into
irreversible ones starting from the corresponding degeneration of equilibrium activities
to zero.

Let us take a system with detailed balance and send some of the equilibrium
activities to zero: $a_i^{\rm eq} \to 0$ when $i \in I$ for some set of indexes
$I$. Immediately we find a surprise: this assumption is not sufficient to find
a limiting irreversible mechanism. It is necessary to take into account the
rates of the convergency to zero of different $a_i^{\rm eq}$. Indeed, let us
study a very simple example, $$A_1\underset{k_{-1}}{
\overset{k_{1}}{\rightleftharpoons}}  A_2 \underset{k_{-2}}{
\overset{k_{2}}{\rightleftharpoons}}  A_3$$ when $a_1^{\rm eq}, a_2^{\rm eq}
\to 0$.

If $a_1^{\rm eq}, a_2^{\rm eq} \to 0$, $a_1^{\rm eq}/ a_2^{\rm eq}=const>0$ and
$a_3^{\rm eq}=const>0$ then the limit system should be $A_1\underset{k_{-1}}{
\overset{k_{1}}{\rightleftharpoons}} A_2 \to  A_3$ and we can keep
$k_{1,-1,2}=const$ whereas $k_{-2}\to 0$.

If $a_1^{\rm eq}, a_2^{\rm eq} \to 0$, $a_1^{\rm eq}/ a_2^{\rm eq}\to 0$ then
the limit system should be $A_1\to  A_2 \to  A_3$ and we can keep
$k_{1,2}=const>0$ whereas $k_{-1,-2}\to 0$.

If $a_1^{\rm eq}, a_2^{\rm eq} \to 0$, $a_2^{\rm eq}/ a_1^{\rm eq}\to 0$ then
in the limit survives only one reaction $A_2 \to  A_3$ (if we assume that all
the reaction rate constants are bounded).

We study asymptotics $a_i^{\rm eq} = {\rm const}\times \varepsilon^{\delta_i}$,
$\varepsilon \to 0$ for various values of non-negative exponents $\delta_i \geq
0$ ($i=1,\ldots,n$). At equilibrium, each reaction rate in the generalized mass
action law is proportional to a power of $\varepsilon$:
$$w_r^{{\rm
eq}+}= k_r^+ {\rm const} \times \varepsilon^{\sum_i \alpha_{ri} \delta_i}\, ,
\;\; w_r^{{\rm eq}-}=k_r^- {\rm const} \times \varepsilon^{\sum_i \beta_{ri}
\delta_i}\, .$$ According to the principle of detailed balance, $w_r^{{\rm
eq}+}=w_r^{{\rm eq}-}$ and
\begin{equation}\label{asymptk+/k-}
\frac{k_r^+}{k_r^-}={\rm const} \times \varepsilon^{(\gamma_r, \delta)}\, ,
\end{equation}
where $\delta$ is the vector of exponents, $\delta=(\delta_i)$.

There are three groups of reactions with respect to the given vector  $\delta$:
$$1. \, (\gamma_r, \delta)=0; \;\; 2. \, (\gamma_r, \delta)<0; \;\; 3.\, (\gamma_r, \delta)>0 \, .$$
In the first group ($(\gamma_r, \delta)=0$) the ratio ${k_r^+}/{k_r^-}$ remains constant
and we can take $k_r^{\pm}= const> 0$. In the second group ($(\gamma_r, \delta)<0$) the
ratio $k_r^{ -}/k_r^{ +}\to 0$ and we should take $k_r^{ -}\to 0$ whereas $k_r^{ +}$ may
remain constant and positive. In the third group ($(\gamma_r, \delta)>0$), the situation
is inverse: $k_r^{+}/k_r^{-}\to 0$ and we can take $k_r^{-} =const>0$, whereas
$k_r^{+}\to 0$.

These three groups depend on $\delta$ but this dependence is piecewise constant. For
every $\gamma_r$, three sets of $\delta$ are defined: (i) hyperplane $(\gamma_r,
\delta)=0$, (ii) hemispace $(\gamma_r, \delta)<0$ and hemispace $(\gamma_r, \delta)>0$.
The space of vectors $\delta$ is split in the subsets defined by the values of functions
${\rm sign} (\gamma_r, \delta)$ ($\pm 1$ or 0).

We consider bounded systems, hence the negative values of $\delta$ should be forbidden.
At least one equilibrium activity should not vanish. Therefore, $\delta_j=0$ for some
$j$. Below we assume that $\delta_i \geq 0$ and $\delta_j=0$ for a non-empty set of
indices $J_0$. Moreover, the atom balance in equilibrium should be positive. Here, this
means that for the set of equilibrium concentrations $c^{\rm eq}_i$ ($i \in J_0$) the
corresponding values of all atomic concentrations are strictly positive and separated
from zero.

Let the vector of exponents, $\delta=(\delta_i)$ be given and the three groups
of reactions be found. For the reactions of the third group (with
$(\gamma_r,\delta)>0$) the forward reaction vanishes in the limit $\varepsilon
\to 0$. It is convenient to transpose the stoichiometric equations for these
reactions and swap the forward reactions with reverse ones. Let us perform this
transposition. After that, $\alpha_r$ changes over $\beta_r$,  $\gamma$
transforms into $-\gamma$, and the inequality $(\gamma_r,\delta)>0$ transforms
into $(\gamma_r,\delta)<0$.

Let us summarize. We use the given vector of exponents $\delta$ and produce a system with
some irreversible reactions from a system of reversible reactions and detailed balance
equilibrium $a_i^{\rm eq}$ by the following rules:
\begin{enumerate}
\item{if $\delta_i >0$ then we assign $a_i^{\rm eq} = 0$ and if $\delta_i =0$ then
    $a_i^{\rm eq}$ does not change;}
\item{if $(\gamma_r, \delta)=0$ then $k_r^{\pm}$ do not change;}
\item{if $(\gamma_r, \delta)<0$ then we assign $k_r^- = 0$ and $k_r^+$ does not
    change;}
\item{if $(\gamma_r, \delta)>0$ then  we assign $k_r^+ = 0$ and $k_r^-$ does not
    change. (In the last case, we transpose the stoichiometric equation and swap the
    forward reaction with reverse one, for convenience, $\gamma_r$ changes to
    -$\gamma_r$ and $k_r^-$ becomes 0. Therefore, this case transforms into case 3.)}
\end{enumerate}
This is a limit system caused by the multiscale degeneration of equilibrium. The
multiscale character of the limit $a_i^{\rm eq} = const \times \varepsilon^{\delta_i} \to
0$ (for some $i$) is important because for different values of $\delta$  reactions may
have different dominant directions and the set of irreversible reactions in the limit may
change.

The general form of the kinetic equations for the homogeneous systems is
\begin{equation}\label{kinur}
\frac{\D N}{\D t}=V \sum_r w_r \gamma_r \, ,
\end{equation}
where $N_i$ is the amount of $A_i$, $N$ is the vector with components $N_i$ and $V$ is
the volume.

Let us consider a limit system for the degeneration of equilibrium with the vector of
exponents $\delta$. For this system $(\gamma_r, \delta)\leq 0$ for all $r$ and, in
particular, $(\gamma_r, \delta) < 0$ for all irreversible reactions and $(\gamma_r,
\delta) = 0$ for all reversible reactions.

\begin{proposition}\label{LyapunovIrrev}
A linear functional $G_{\delta}(N)=(\delta,N)$ decreases along the solutions of
kinetic equations (\ref{kinur}) for this limit system: $\D G_{\delta}(N)/{\D
t}\leq 0$ and $\D G_{\delta}(N){\D t}= 0$ if and only if all the reaction rates
for the irreversible reactions are zero.
\end{proposition}
\begin{proof} Indeed,
\begin{equation}\label{LinearDiss}
\frac{\D G_{\delta}(N)}{\D t}=V \sum_r w_r (\gamma_r,\delta) \leq 0 \, ,
\end{equation}
because for reversible reactions $(\gamma_r, \delta) = 0$, and for irreversible
reactions $w_r=w_r^+\geq 0$ and $(\gamma_r, \delta) < 0$. All the terms in this
sum are non-negative, hence it may be zero if and only if each summand is zero.
\end{proof}

This Lyapunov function may be used in a proof that the rates of all irreversible
reactions in the system tend to 0 with time. Indeed, if they do not tend to zero then on
a solution of (\ref{kinur}) $G_{\delta}(N(t)) \to -\infty$ when $t\to \infty$ and $N(t)$
is unbounded. Equation (\ref{LinearDiss}) and Proposition~(\ref{LyapunovIrrev}) give us
the possibility to prove the extended principle of detailed balance in the following
form. Let us consider a reaction mechanism that includes reversible and irreversible
reactions. Assume that the reaction rates satisfy the generalized mass action law
(\ref{GenMAL}) and the set of reaction rate constants is given. Let us ask the question:
Is it possible to obtain this reaction mechanism and reaction rate constants as a limit
in the multiscale degeneration of equilibrium from a fully reversible system with the
classical detailed balance. The answer to this question gives the following theorem about
the extended principle of detailed balance.

\begin{theorem}\label{Theorem:ExPrDetBalDegenerat}
A system can be obtained as a limit in the multiscale degeneration of equilibrium from a
reversible system with detailed balance if and only if (i) the reaction rate constants of
the reversible part of the reaction mechanism satisfy the classical principle of detailed
balance and (ii) the convex hull of the stoichiometric vectors of the irreversible
reactions does not intersect the linear span of the stoichiometric vectors of reversible
reactions.
\end{theorem}
\begin{proof}
Let the given system be a limit  of a reversible system with detailed balance in the
multiscale degeneration of equilibrium with the exponent vector $\delta$. Then for the
reversible reactions $(\gamma_r,\delta)=0$ and for the irreversible reactions $(\gamma_r,
\delta) < 0$. For every vector $x$ from the convex hull of the stoichiometric vector of
the irreversible reactions $(x, \delta) < 0$ and for any vector $y$ from the linear span
of the stoichiometric vectors of the reversible reactions $(y,\delta)=0$. Therefore,
these sets do not intersect. The reaction rate constants for the reversible reactions
satisfy the classical principle of detailed balance because they do not change in the
equilibrium degeneration and keep this property of the original fully reverse system with
detailed balance.

Conversely, let a system satisfy the extended principle of detailed balance:  (i) the
reaction rate constants of the reversible part of the reaction mechanism satisfy the
classical principle of detailed balance and (ii) the convex hull of the stoichiometric
vectors of the irreversible reactions does not intersect the linear span of the
stoichiometric vectors of reversible reactions. Due to the classical theorems of the
convex geometry, there exists a linear functional that separates this convex set from the
linear subspace. (Strong separation of closed and compact convex sets.) This separating
functional can be represented in the form $(x,\theta)$ for some vector $\theta$. For the
reversible reactions $(\gamma_r,\theta)=0$ and for the irreversible reactions $(\gamma_r,
\theta) < 0$.

It is possible to find vector $\delta$ with this separation property and non-negative
coordinates. Indeed, according to the basic assumptions, there exists a linear
conservation law with strongly positive coordinates. This is a vector $b$ ($b_i>0$) with
the property: $(\gamma_r,b)=0$ for all reactions. For any $\lambda$, the vector
$\theta+\lambda b$ has the same separation property as the vector $\theta$ has. We can
select such $\lambda$ that $ \delta_i=\theta_i+\lambda b_i \geq 0$ and
$\delta_i=\theta_i+\lambda b_i = 0$ for some $i$. Let us take this linear combination
$\delta$ as a vector of exponents.

Let us create a fully reversible system from the initial partially irreversible one. We
do not change the reversible reactions and their rate constants. Because the reversible
reactions satisfy the classical principle of detailed balance, there exists a strongly
positive vector of equilibrium activities $a_i^*>0$ for the reversible reactions.

For each irreversible reaction with the stoichiometric vector $\gamma_r$ and reaction
rate constant $k_r=k_r^+>0$ we add a reverse reaction with the reaction rate constant
$$k_r^-=k_r^+\prod_i(a_i^*)^{-\gamma_{ri}}\, .$$
For this fully reversible system the activities $a_i^*>0$ provide the point of detailed
balance. In the multiscale degeneration process, the equilibrium activities depend on
$\varepsilon \to 0$ as $a_i^{\rm eq}=a_i^* \varepsilon^{\delta_i}$. For the reactions
with $(\gamma_r,\delta)=0$ the reaction rate constants do not depend on $\varepsilon$ and
for the reactions with $(\gamma_r,\delta)<0$ the rate constant $k_r^-$ tends to zero as
$\varepsilon^{-(\gamma_r,\delta)}$ and $k_r^+$ does not change. We return to the initial
system of reactions in the limit $\varepsilon \to 0$.
\end{proof}

This is a particular form of the extended principle of detailed balance. For more
discussion see \cite{GorbYabCES2012}. Fig.~\ref{Fig1ad} illustrates geometric sense of
the extended detailed balance condition: the convex hull of the stoichiometric vectors of
the irreversible reactions does not intersect the linear span of the stoichiometric
vectors of the reversible reactions. In this illustration, $\{\gamma_r\, | \, r\in J_0\}$
are the stoichiometric vectors of the reversible reactions and $\{\gamma_r\, | \, r\in
J_1\}$ are the stoichiometric vectors of the irreversible reactions.

\begin{figure}
\centering{
\includegraphics[width=0.38\textwidth]{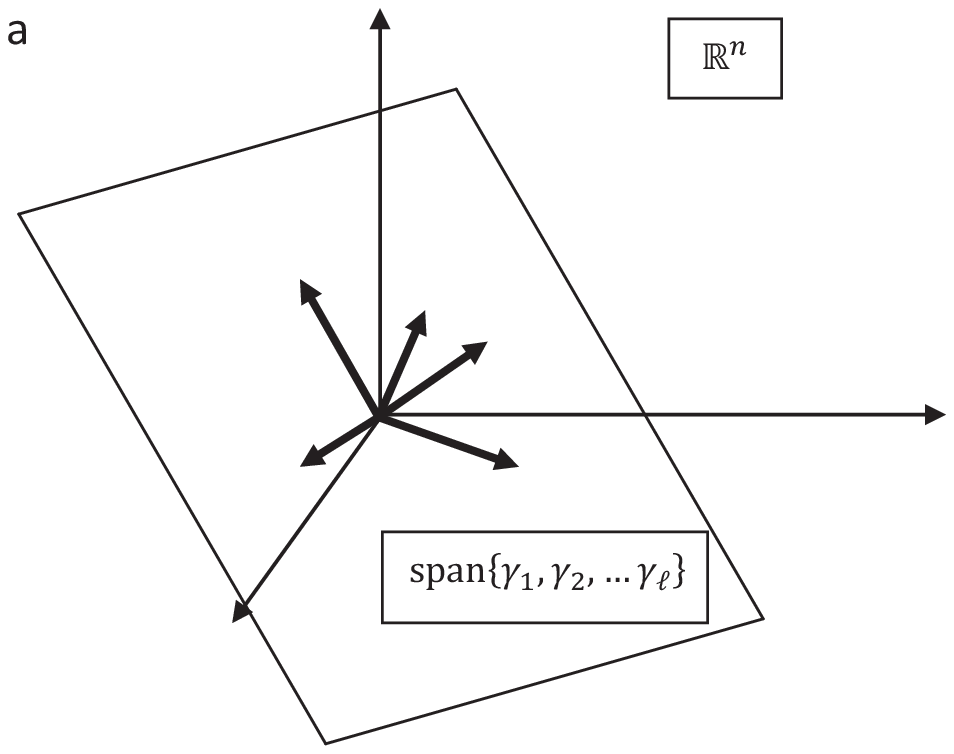}\quad
\includegraphics[width=0.4\textwidth]{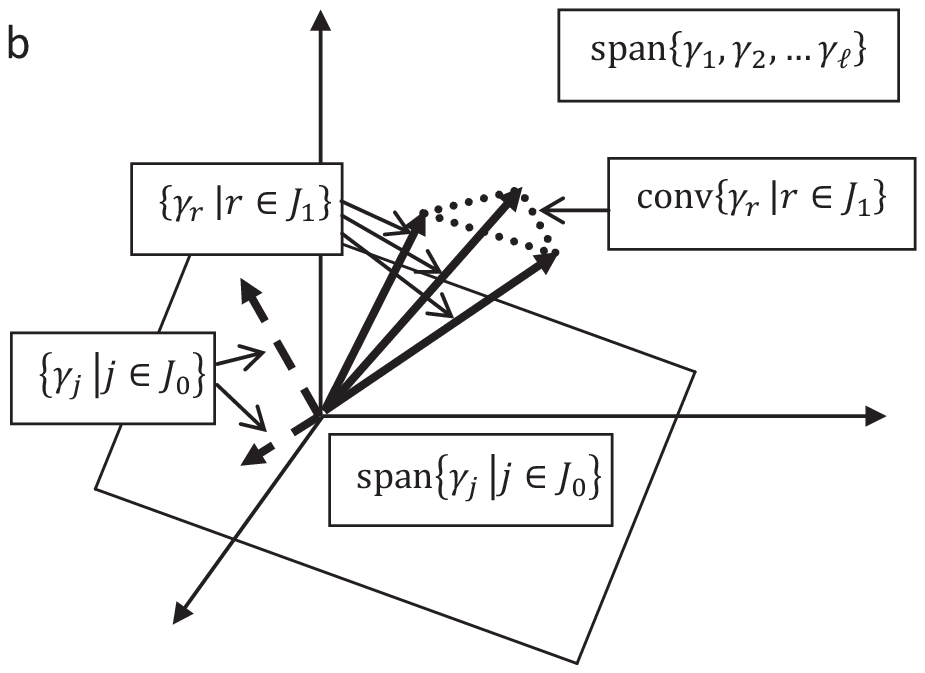} \\
\includegraphics[width=0.35\textwidth]{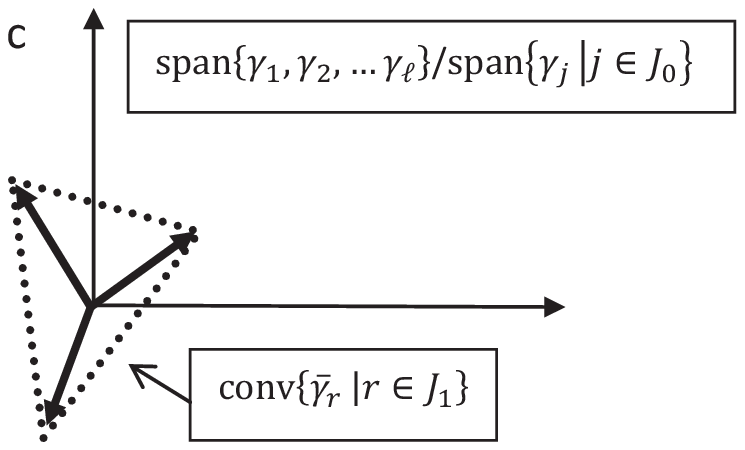} \qquad
\includegraphics[width=0.4\textwidth]{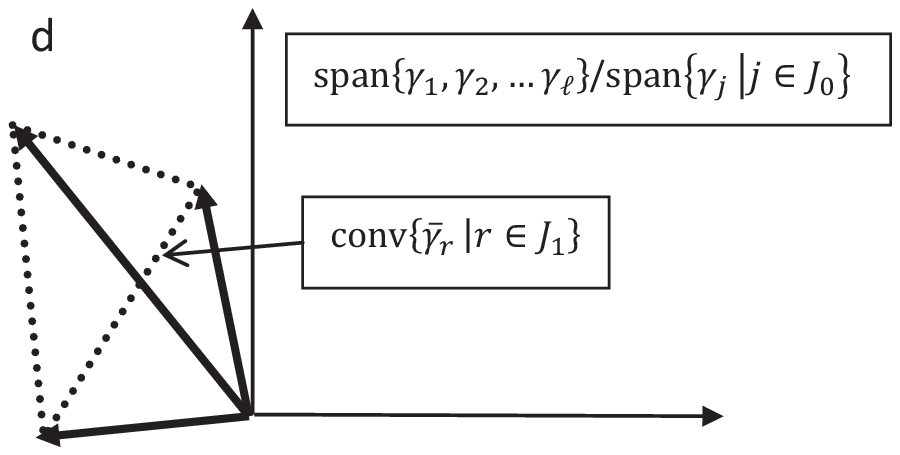}
\caption{\label{Fig1ad}Main operations in the application of the extended detailed balance conditions. In the concentration space $\mathbb{R}^n$
we should find the subspace spanned by all the stoichiometric vectors $\{\gamma_r\, | \, r=1,  \ldots , {\ell}\}$ (a).
In this subspace we have to select the internal coordinates. In span$\{\gamma_r\, | \, r=1,  \ldots , {\ell}\}$ we have to select
the subspace spanned by the stoichiometric vectors of the reversible reactions (b) (the dashed vectors).
The stoichiometric vectors of the irreversible reactions are in (b,c,d) solid and bold.
Due to the extended principle of detailed balance,
span$\{\gamma_r\, | \, r\in J_0\}$ should not intersect conv$\{\gamma_r\, | \, r\in J_1\}$ (the dotted triangle in Fig.).
For analysis of this intersection, it is convenient to  proceed to the quotation space
span$\{\gamma_r\, | \, r=1,  \ldots , {\ell}\}$/span$\{\gamma_r\, | \, r\in J_0\}$ (c,d). In this quotation space,
span$\{\bar{\gamma}_r\, | \, r\in J_0\}$ is $\{0\}$ and two situations are possible: (c) $0\in \mbox{conv}\{\bar{\gamma}_r\, | \, r\in J_1\}$
(the dotted triangle includes the origin) or (d) $0\notin \mbox{conv}\{\gamma_r\, | \, r\in J_1\}$
(the dotted triangle does not include the origin). In the case (c) the extended detailed balance condition is violated. The case (d) satisfies this condition.}}
\end{figure}

\subsection{Activities, concentrations and affinities \label{SecActConc}}

To combine the linear Lyapunov functions $G_{\delta}(N)=(\delta,N)$
(\ref{LinearDiss}) with the classical thermodynamic potential and study the
kinetic equations in the closed form we have to specify the relations between
activities and concentrations. We accept the assumption: $a_i=c_i g_i(c,T) $,
where $g_i(c,T)>0$ is the activity coefficient. It is  a continuously
differentiable function of $c,T$ in the whole diapason of their values. In a
bounded region of concentrations and temperature we can always assume that
$g_i>g_0>0$ for some constant $g_0$. This assumption is valid for the non-ideal
gases and for liquid solutions. It holds also for the ``surface gas" in
kinetics of heterogeneous catalysis \cite{Yablonskiiatal1991} and does not hold
for the solid reagents (see for example, analysis of carbon activity in the
methane reforming \cite{GorbYabCES2012}).

The system of units should be commented. Traditionally, $a_i$ is assumed to be
dimensionless and for perfect systems $a_i=c_i/c_i^{\circ}$, where
$c_i^{\circ}$ is an arbitrary ``standard" concentration. To avoid introduction
of unnecessary quantities, we always assume that in the selected system of
units, $c_i^{\circ}\equiv 1$.

If the thermodynamic potentials exist then due to the thermodynamic  definition
of activity (\ref{StandardActivity}), this hypothesis is equivalent to the
logarithmic singularity of the chemical potentials, $\mu_i= RT\ln c_i +\ldots$
where $\ldots$ stands for a continuous function of $c,T$ (all the
concentrations and the temperature). In this case,  the free energy has the
form
\begin{equation}\label{FreeEnLogSIng}
F(N,T,V)=RT\sum_i N_i (\ln c_i-1 + f_{0i}(c,T))\, ,
\end{equation}
where the functions $f_{0i}(c,T)$ are continuously differentiable for all
possible values of arguments. Functions $f_{0i}$ in the right hand side of the
representation (\ref{FreeEnLogSIng}) cannot be restored unambiguously from the
free energy function $F(N,T,V)$ but for a small admixture $A_i$ it is possible
to introduce the partial pressure $p_i$ which satisfies the law
$p_i=RTc_i+o(c_i)$. This is due to the terms $N_i\ln c_i$ in $F$. Indeed,
$P=-\partial F(N,T,V)/\partial V=RTc_i+o(c_i)+ P\left|_{c_i=0} \right. $.
Connections between the equation of state, free energy and kinetics are
discussed in more detail in \cite{GiovangigliMatus2012,G11984}.

There are several simple algebraic corollaries of the assumed connection
between activities and concentrations. Let us consider an elementary reaction
$\sum \alpha_i A_i \to \sum \beta_i A_i$ with $\alpha_i, \beta_i \geq 0$. Then,
according to the generalized mass action law, for any vector of concentrations
$c$ ($c_i \geq 0$)
\begin{enumerate}
\item{If, for some $i$, $c_i =0$ then $\gamma_i w(c) \geq 0$;}
\item{If, for some $i$, $c_i =0$ and $\gamma_i <0$ then $\alpha_i > 0$ and $w(c)=0$.}
\end{enumerate}
Similarly, for a reversible reaction $\sum \alpha_i A_i \rightleftharpoons \sum
\beta_i A_i$
\begin{enumerate}
\item{If, for some $i$, $c_i =0$ and $\gamma_i >0$ then $\beta_i > 0$ and $w^-(c)=0$;}
\item{If, for some $i$, $c_i =0$ and $\gamma_i <0$ then $\alpha_i > 0$ and $w^+(c)=0$.}
\end{enumerate}
These statements as well as Proposition~\ref{BorderZeros} and
Corollary~\ref{Cor:FaceInvariance} below are the consequences of the generalized mass
action law (\ref{GenMAL}) and the connection between activities and concentrations
without any assumptions about extended principle of detailed balance.

Each set of indexes $J=\{i_1,\ldots , i_j\}$ defines a face of the positive polyhedron,
$$F_J=\{c\, | \, c_i \geq 0 \mbox{ for all } i \mbox{ and } c_i=0 \mbox{ for } i\in J \}\, .$$
By definition, the relative interior of $F_J$, $ri(F_J)$, consists of points with $c_i=0$
for $i\in J$ and $c_i >0$ for $i \notin J$.

\begin{proposition}\label{BorderZeros}Let for a point $c \in ri(F_J)$ and an index $i\in J$
$$\sum_r \gamma_{ri} w_r(c) = 0\, .$$
Then this identity holds for all $c \in F_J$.
\end{proposition}
\begin{proof}For convenience, let us write all the forward and reverse reactions
separately and represent the reaction mechanism in the form (\ref{ReactMech}).
All the terms in the sum $\sum_r \gamma_{ri} w_r(c)$ are non-negative, because
$c_i=0$. Therefore, if the sum is zero then all the terms are zero. The
reaction rate $w_r$ (\ref{GenMAL}) with non-zero rate constant takes zero value
if and only if $\alpha_{rj}>0$ and $a_j=0$ for some $j$. The equality $a_i=0$
is equivalent to $c_i=0$. Therefore, $w_r(c)=0$ for a point $c\in ri(F_J)$ if
and only if there exists $j\in J$ such that $\alpha_{rj}>0$. If $\alpha_{rj}>0$
for an index $j\in J$ then $w_r(c)=0$ for all $c \in F_J$ because $c_j=0$ in
$F_J$.
\end{proof}

We call a face $F_J$ of the positive orthant $\mathbb{R}^n_+$ {\em invariant} with
respect to a set $S$ of elementary reactions if $\sum_{r\in S} \gamma_{rj} w_r(c) = 0\, $
for all $c\in F_J$ and every $j \in J$.

Let us consider the reaction mechanism in the form (\ref{ReactMech}) where all
the forward and reverse reactions participate separately.
\begin{corollary}\label{Cor:FaceInvariance}
The following statements are equivalent:
\begin{enumerate}
\item{$\sum_{r\in S} \gamma_{ri} w_r(c) = 0$ for a point $c \in ri(F_J)$ and all
    indexes $i\in J$;}
\item{The face $F_J$ is invariant with respect to the set of reactions $S$;}
\item{The face $F_J$ is invariant with respect to every elementary reaction from
    $S$;}
\item{For every $r\in S$ either $\gamma_{rj}=0$ for all $j \in J$ or $\alpha_{rj}>0$
    for some $j\in J$.}
\end{enumerate}
\end{corollary}

We aim to perform the analysis of the asymptotic behavior of the kinetic
equations in the multiscale degeneration of equilibrium described in
Sec.~\ref{Sec:Shifted}. For this purpose, we have to answer the question: how
the relations between activities $a_i$ and concentrations $c_i$ depend on the
degeneration parameter $\varepsilon \to 0$? We do no try to find the maximally
general appropriate answer to this question. For the known applications, the
answer is: the relations between $a_i$ and $c_i$ do not depend on $\varepsilon
\to 0$. In particular, it is trivially true for the ideal systems. The simple
generalization, $a_i=c_i g_i(c,T,\varepsilon) $, where
$g_i(c,T,\varepsilon)>g_0>0$ are continuous functions, is not a generalization
at all, because we can use for $\varepsilon \to 0$ the limit case that does not
depend on $\varepsilon$, $g_i(c,T)=g_i(c,T,0)$.

This independence from $\varepsilon$ implies that the reversible part of the reaction
mechanism has the thermodynamic Lyapunov functions like free energy. If we just delete
the irreversible part then the classical thermodynamics is applicable and the
thermodynamic potentials do not depend on $\varepsilon$. For the generalized mass action
law, the time derivative of the relevant thermodynamic potentials have very nice general
form. Let, under given condition, the function $\Phi(N,\ldots)$ be given, where by
$\ldots$ is used for the  quantities that do not change in time under these conditions.
It is the thermodynamics potential if $\partial \Phi(N,\ldots)/\partial N_i =\mu_i$. For
example, it is the free Helmholtz energy $F$ for $V,T=const$ and the free Gibbs energy
$G$ for $P,V=const$.

Let us calculate the time derivative of $\Phi(N,\ldots)$ due to kinetic equation
(\ref{kinur}). The reaction rates are given by the generalized mass action law
(\ref{GenMAL}) with definition of activities through chemical potential
(\ref{StandardActivity}). We assume that the principle of detailed balance holds (it
should hold for the reversible part of the reaction mechanism according to the extended
detailed balance conditions). More precisely, there exists an equilibrium with detailed
balance for any temperature $T$, $a^{\rm eq}(T)$: for all $r$, $w_r^+(a^{\rm
eq})=w_r^-(a^{\rm eq})=w_r^{\rm eq}(T)$. It is convenient to represent the reaction rates
using these {\em equilibrium fluxes} $w_r^{\rm eq}(T)$:
\begin{equation*}
w^+_r=w^{\rm eq}_r \exp \left(\sum_i \frac{\alpha_{ri}(\mu_i-\mu^{\rm
eq}_i)}{RT}\right)\, , \;\; w^-_r=w^{\rm eq}_r \exp \left(\sum_i
\frac{\beta_{ri}(\mu_i-\mu^{\rm eq}_i)}{RT}\right)\, .
\end{equation*}
where $\mu^{\rm eq}_i=\mu_i(a^{\rm eq},T)$.

These formulas give immediately the following representation of the dissipation rate
\begin{equation}\label{DissClass}
\begin{split}
\frac{\D \Phi}{\D t}&=\sum_i \frac{\partial \Phi(N,\ldots)}{\partial N_i} \frac{\D N_i}{\D
t}=\sum_i \mu_i \frac{\D N_i}{\D t} \\&= -VRT \sum_r (\ln w_r^+-\ln w_r^-) (w_r^+-w_r^-) \leq
0\, .
\end{split}
\end{equation}
The inequality holds because $\ln$ is a monotone function and, hence, the expressions
$\ln w_r^+-\ln w_r^-$ and $w_r^+-w_r^-$ have always the same sign. Formulas of this kind
for dissipation are well known since the famous Boltzmann $H$-theorem (1873
\cite{Boltzmann}, see also \cite{Grmela2010}). The entropy increase in isolated systems
has the similar form:
$$
\frac{\D S}{\D t}=VR \sum_r (\ln w_r^+-\ln w_r^-) (w_r^+-w_r^-)\geq 0 \, .
$$
Let us notice that
$$
\ln w_r^+-\ln w_r^-=\frac{1}{RT}\sum_i\mu_i(\alpha_{ri}-\beta_{ri})=-\frac{(\gamma_r,\mu)}{RT}\, .
$$
The quantity $-(\gamma_r,\mu)$ is one of the central notion of physical
chemistry, {\em affinity} \cite{de Donder1936}. It is positive if the forward
reaction prevails over reverse one and negative in the opposite case. It
measures the energetic advantage of the forward reaction over the reverse one
(free energy per mole). The activity divided by $RT$ shows how large is this
energetic advantage comparing to the thermal energy. We call it the {\em
normalized affinity} and use a special notation for this quantity:
$$\mathbb{A}_r=-\frac{(\gamma_r,\mu)}{RT}$$
Let us apply an elementary identity
$$\exp a - \exp b = (\exp a + \exp b) \tanh \frac{a-b}{2}$$
to the reaction rate, $w_r=w_r^+-w_r^-$:
\begin{equation}\label{RatesThroughAff}
w_r=(w_r^+ + w_r^-) \tanh \frac{\mathbb{A}_r}{2}\, .
\end{equation}
This representation of the reaction rates gives immediately for the dissipation rate:
\begin{equation}\label{DissTanh}
\frac{\D \Phi}{\D t}=-VRT \sum_r (w_r^+ + w_r^-)\mathbb{A}_r  \tanh \frac{\mathbb{A}_r}{2} \leq
0\, .
\end{equation}
In this formula, the kinetic information is collected in the non-negative factors, the
sums of reaction rates $(w_r^+ + w_r^-)$. The purely thermodynamical multipliers
$\mathbb{A}_r \tanh ({\mathbb{A}_r}/{2})$ are also non-negative. For small
$|\mathbb{A}_r|$, the expression $\mathbb{A}_r  \tanh ({\mathbb{A}_r}/{2})$ behaves like
$\mathbb{A}_r^2/2$ and for large $|\mathbb{A}_r|$ it behaves like the absolute value,
$|\mathbb{A}_r|$.

So, we have two Lyapunov functions for two fragments of the reaction mechanism. For the
reversible part, this is just a classical thermodynamic potential. For the irreversible
part, this is a linear functional $G_{\delta}(N)=(\delta,N)$. More precisely, the
irreversible reactions decrease this functional, whereas for the reversible reactions it
is the conservation law. Therefore, it decreases monotonically in time for the whole
system.

\section{Attractors \label{Sec:Attr}}

\subsection{Dynamical systems and limit points}

The kinetic equations (\ref{kinur}) do not give a complete representation of
dynamics. The right hand side includes the volume $V$ and the reaction rates
$w_r$ which are functions (\ref{GenMAL}) of the concentrations $c$ and
temperature $T$, whereas in the left hand side there is $\dot{N}$. To close
this system, we need to express $V$, $c$ and $T$ through $N$ and quantities
which do not change in time. This closure depends on conditions. The simplest
expressions appear for isochoric isothermal conditions: $V,T=const$, $c=N/V$.
For other classical conditions ($U,V=const$, or $P,T=const$, or $H,P=const$) we
have to use the equations of state. There may be more sophisticated closures
which include models or external regulators of the pressure and temperature,
for example.

Proposition~\ref{LyapunovIrrev} is valid for all possible closures. It is only important
that the external flux of the chemical components is absent. Further on, we assume that
the conditions are selected, the closure is done, the right hand side of the resulting
system is continuously differentiable  and there exists the positive bounded solution for
initial data in $\mathbb{R}^n_+$ and $V$, $T$ remain bounded and  separated from zero.
The nature of this closure is not crucial. For some important particular closures the
proofs of existence of positive and bounded solutions are well known (see, for example,
\cite{VolpertKhudyaev1985}). Strictly speaking, such a system is not a dynamical system
in $\mathbb{R}^n_+$ but a semi-dynamical one: the solutions may lose positivity and leave
$\mathbb{R}^n_+$ for negative values of time. The theory of the limit behavior of the
semi-dynamical systems was developed for applications to kinetic systems
\cite{GorbanEJDE2004}.

We aim to describe the limit behavior of the systems as $t\to \infty$. Under
the extended detailed balance condition the limit behavior is rather simple and
the system will approach steady states but to prove this behavior we need the
more general notion of the $\omega$-limit points.

By the definition, the $\omega$-limit points of a dynamical system are the
limit points of the motions when time $t \to \infty$. We consider a kinetic
system in $\mathbb{R}^n_+$.  In particular, for each solution of the kinetic
equations $N(t)$ the set of the corresponding $\omega$-limit points is closed,
connected and consists of the whole trajectories (\cite{GorbanEJDE2004},
Proposition 1.5). This means that the motion which starts from an
$\omega$-limit point remains in $\mathbb{R}^n_+$ for all time moments, both
positive and negative.

\begin{proposition}\label{limitpoints}Let $N(t)$ be a positive solution of the kinetic
equation and $x^*$ be an $\omega$-limit point of this solution and $x_i^*=0$.
then at this point $\dot{x}_i|_{x^*}=0$.
\end{proposition}
\begin{proof} Let $x(t)$ be a solution of the kinetic equations with the initial state
$x(0)=x^*$. All the points $x(t)$ ($-\infty <t<\infty$) belong to
$\mathbb{R}^n_+$. Indeed, there exists such a sequence $t_j \to \infty$ that
$N(t_j) \to x^*$. For any $\tau \in (-\infty,\infty)$, $N(t_j+\tau) \to
x(\tau)$. For sufficiently large $j$, $t_j+\tau >0$ and the value $N(t_j+\tau)
\in \mathbb{R}^n_+ $. Therefore, $x(\tau)\in \mathbb{R}^n_+$ ($-\infty
<\tau<\infty$) and for any $\tau$ the point $x(\tau)$ is an $\omega$-limit
point of the solution $N(t)$. Let $x_i^*=0$ and $\dot{x}_i|_{x^*}=v \neq 0$. If
$v>0$ then for small $|\tau|$ and $\tau<0$ the value of $x_i$ becomes negative,
$x_i(\tau)<0$. It is impossible because positivity. Similarly, If $v<0$ then
for small $\tau >0$ the value of $x_i$ becomes negative, $x_i(\tau)<0$. It is
also  impossible because positivity. Therefore, $\dot{x}_i|_{x^*}=0$.
\end{proof}

We use Proposition~\ref{limitpoints} in the following combination with
Proposition~\ref{BorderZeros}. Let us write the reaction mechanism in the form
(\ref{ReactMech}).
\begin{corollary}\label{Corol:FaceLimit}
If an $\omega$-limit point belongs to the relative interior $riF_J$ of the face
$F_J \subset \mathbb{R}^n_+$  then the face $F_J$ is invariant with respect to
the reaction mechanism and for every elementary reaction either $\gamma_{rj}=0$
for all $j \in J$ or $\alpha_{rj}>0$ for some $j\in J$.
\end{corollary}
\begin{proof}If an $\omega$-limit point belongs to $riF_J$
then at this point all $\dot{c}_j=0$ for $j\in J$
due to Proposition~\ref{limitpoints}. Therefore, we can apply
Corollary~\ref{Cor:FaceInvariance}.\end{proof}

\subsection{Steady states of irreversible reactions \label{Sec:StSt}}

Under extended detailed balance conditions, all the reaction rates of the
irreversible reactions are zero at every limit point of the kinetic equations
(\ref{kinur}), due to Proposition~\ref{LyapunovIrrev}. In this section, we give
a simple combinatorial description of steady states for the set of irreversible
reactions. This description is based on Proposition~\ref{LyapunovIrrev} and,
therefore, uses the extended detailed balance conditions.

We continue to study multiscale degeneration of a detailed balance equilibrium.
The vector of exponents $\delta=(\delta_i)$  is given, $\delta_i \geq 0$ for
all $i$ and $\delta_i = 0$ for some $i$. There are two sets of reaction. For
one of them, $(\gamma_r, \delta)=0$ and in the limit both $k_r^{\pm}>0$. In the
second set, $(\gamma_r, \delta)<0$ and in the limit we assign $k_r^- = 0$ and
$k_r^+$ is the same as in the initial system (before the equilibrium
degeneration). If it is necessary, we transpose the stoichiometric equations
and swap the forward reactions with reverse ones.

For convenience, let us change the notations. Let $\gamma_i$ be the
stoichiometric vectors of reversible reactions with $(\gamma_r, \delta)=0$
($r=1, \ldots , h$), and $\nu_l$  be the stoichiometric vectors for the
reactions from the second set, $(\nu_l, \delta)<0$ ($l=1, \ldots , s$). For the
reaction rates and constants for the first set we keep the same notations:
$w_r$, $w_r^{\pm}$, $k_r^{\pm}$. For the second set, we use for the reaction
rate constants $q_l=q_l^+$ and for the reaction rates $v_l=v_l^+$. (They are
also calculated according to the generalized mass action law (\ref{GenMAL}).)
The input and output stoichiometric coefficients remain $\alpha_{ri}$ and
$\beta_{ri}$ for the first set and for the second set we use the notations
$\alpha_{li}^{\nu}$ and $\beta_{li}^{\nu}$.

Let the rates of all the irreversible reaction be equal to zero. This does not mean that
all the concentrations $a_i$ with $\delta_i>0$ achieve zero. A bimolecular reaction
$A+B\to C$ gives us a simple example: $w=k a_A a_B$ and $w=0$ if either $a_A=0$ or
$a_B=0$. On the plane with coordinates $a_A, a_B$ and with the positivity condition,
$a_A, a_B \geq 0$, the set of zeros of $w$ is a union of two semi-axes, $\{a_A=0, a_B
\geq 0\}$ and $\{a_A\geq 0, a_B = 0\}$. In more general situation, the set in the
activity space, where all the irreversible reactions have zero rates, has a similar
structure: it is the union of some faces of the  positive orthant.

Let us describe the set of the steady states of the irreversible reactions. Due
to Proposition~\ref{LyapunovIrrev}, if $\sum_l v_l \nu_l=0$ then all $v_l=0$.
Let us describe the set of zeros of all $v_l$ in the the positive orthant of
activities.

For every $l=1, \ldots , s$ the set of zeros of $v_l$ in  $\mathbb{R}^n_+$ is
given by the conditions: at least for one $i$ $\alpha_{li}^{\nu}\neq 0$ and
$a_i=0$. It is convenient to represent this condition as a disjunction. Let
$J_l = \{i \, | \, \alpha_{li}^{\nu}\neq 0\}$. Then the set of zeros of $v_l$
an a positive orthant of activities is presented by the formula $\bigvee_{i \in
J_l} (a_i=0)$. The set of zeros of all $v_l$ is represented by the following
conjunction form
\begin{equation}\label{COnjFormStSt}
\wedge_{l=1}^s \left(\vee_{i \in J_l} (a_i=0)\right) \, .
\end{equation}
To transform it into the unions of subspaces we have to move to a disjunction
form and make some cancelations. First of all, we represent this formula as a
disjunction of conjunctions:
\begin{equation}\label{DisForm}
\wedge_{l=1}^s \left(\vee_{i \in J_l}(a_i=0)\right)= \vee_{i_1\in J_1, \ldots ,
i_s \in J_s} \left((a_{i_1}=0) \wedge \ldots \wedge (a_{i_s}=0)\right)\, .
\end{equation}
For a cortege of indexes $\{i_1,\ldots , i_s \}$ the correspondent set of their
values may be smaller because some values $i_l$ may coincide. Let this set of
values be $S_{\{i_1,\ldots , i_s \}}$. We can delete from (\ref{DisForm}) a
conjunction $(a_{i_1}=0) \wedge \ldots \wedge (a_{i_s}=0)$ if there exists a
cortege $\{i'_1,\ldots , i'_s \}$ ($i'_l \in J_l$) with smaller set of values,
$S_{\{i_1,\ldots , i_s \}} \supseteq  S_{\{i'_1,\ldots , i'_s \}}$. Let us
check the corteges in some order and delete a conjunction from (\ref{DisForm})
if there remain a term with smaller (or the same) set of index values in the
formula. We can also substitute in (\ref{DisForm}) the corteges by their sets
of values. The resulting minimized formula may become shorter. Each elementary
conjunction represents a coordinate subspace and after cancelations each this
subspace does not belong to a union of other subspaces. The final form of
formula (\ref{DisForm}) is
\begin{equation}\label{DisForm1}
\vee_j
(\wedge_{i\in S_j} (a_i=0) )\, ,
\end{equation}
where $S_j$  are sets of indexes, $S_j \subset \{1, \ldots ,n\}$ and for every two
different $S_j$, $S_p$ none of them includes another, $S_j \nsubseteq S_p$. The
elementary conjunction $\wedge_{i\in S_j} (a_i=0)$ describes a subspace.

The steady states of the irreversible part of the reaction mechanism are given by the
intersection of the union of the coordinate subspaces (\ref{DisForm1}) with
$\mathbb{R}^n_+$. For applications of this formula, it is important that the equalities
$a_i=0$, $c_i=0$ and $N_i=0$ are equivalent and the positive orthants of the activities
$a_i$,  concentrations $c_i$ or amounts $N_i$ represent the same sets of physical states.
This is also true for the faces of these orthants: $F_J$ for the activities,
concentrations or amounts correspond to the same sets of states. (The same state may
corresponds to the different points of these cones, but the totalities of the states are
the same.)

\subsection{Sets of steady states  of irreversible reactions invariant with respect to
reversible reactions \label{Sec:InvStSt}}

In this Sec. we study the possible limit behavior of systems which satisfy the
extended detailed balance conditions and include some irreversible reactions.
All the $\omega$-limit points of such systems are steady states of the
irreversible reactions due to Proposition~\ref{LyapunovIrrev} but not all these
steady states may be the $\omega$-limit points of the system. A simple formal
example gives us the couple of reaction: $A \rightleftharpoons B$, $A+B \to C$.
Here, we have one reversible and one irreversible reaction. The conditions of
the extended detailed balance hold (trivially): the linear span of the
stoichiometric vector of the reversible reaction, $(-1,1,0)$, does not include
the stoichiometric vector of the irreversible reaction, $(-1,-1,1)$. For the
description of the multiscale degeneration of equilibrium, we can take the
exponents $\delta_A=1,\delta_B=1, \delta_C=0$.

The steady states of the irreversible reaction are given in $\mathbb{R}^n_+$ by
the disjunction, $(c_A=0)\vee (c_B=0)$ but only the points $(c_A=c_B=0)$ may be
the limit points when $t \to \infty$. Indeed, if $c_A=0$ and $c_B >0$ then $\D
c_A /\D t = k_{1}^-c_B>0$. Due to Proposition~\ref{limitpoints} this is not an
$\omega$-limit point. Similarly, the points with $c_A>0$ and $c_B=0$ are not
the $\omega$-limit points.

Let us combine Propositions~\ref{LyapunovIrrev}, \ref{limitpoints} and
Corollary~\ref{Corol:FaceLimit} in the following statement.

\begin{theorem}\label{LocationOmegaLimit}Let the kinetic system satisfy the extended detailed balance
conditions and include some irreversible reactions. Then an $\omega$-limit
point $x^* \in riF_J$ exists if and only if $F_J$ consists of steady states of
the irreversible reactions and is invariant with respect to all reversible
reactions.
\end{theorem}
\begin{proof}If an $\omega$-limit point $x^* \in riF_J$ exists then it is a
steady state for all irreversible reactions (due to
Propositions~\ref{LyapunovIrrev}). Therefore, the face $F_J$ consists of
steady-states of the irreversible reactions (Proposition~\ref{limitpoints}) and
is invariant with respect to all reversible reactions
(Proposition~\ref{limitpoints} and Corollary~\ref{Corol:FaceLimit}). To prove
the reverse statement, let us assume that $F_J$ consists of steady states of
the irreversible reactions and is invariant with respect to all reversible
reactions. The reversible reactions which do not act on $c_j$ for $j\in J$
define a semi-dynamical system on $F_J$. The positive conservation law $b$
defines an positively invariant polyhedron in $F_J$. Dynamics in such a compact
set always has $\omega$-limit points.\end{proof}

Let us find the faces $F_J$ that contain the $\omega$-limit points in their
relative interior $riF_J$. According to Theorem~\ref{LocationOmegaLimit}, these
faces should consist of the steady states of the irreversible reactions and
should be invariant with respect to all reversible reactions. Let us look for
the {\em maximal faces with this property}. For this purpose, we always
minimize the disjunctive forms by cancelations. We do not list the faces that
contain the $\omega$-limit points in their relative interior and are the proper
subsets of other faces with this property. All the $\omega$-limit points belong
to the union of these maximal faces.

Let us start from the minimized disjunctive form (\ref{DisForm1}). Equation
(\ref{DisForm1}) represents the set of the steady states of the irreversible part of the
reaction mechanism by a union of the coordinate subspaces $\wedge_{i\in S_j} (c_i=0)$ in
intersection with $\mathbb{R}^n_+$. It is the union of the faces, $\cup_j F_{S_j}$. If a
face $F_J$ consists of the steady states of the irreversible reactions then $J \supseteq
S_j$ for some $j$.

The following formula (\ref{LogicLimitFace}) is true on a face $F_J$ if it contains
$\omega$-limit points in the relative interior $riF_J$
(Theorem~\ref{LocationOmegaLimit}):
\begin{equation}\label{LogicLimitFace}
(c_i=0) \Rightarrow \left[\left( \wedge_{r, \gamma_{ri}>0}\ \vee_{j,
\alpha_{rj}>0} (c_j=0)\right) \wedge \left( \wedge_{r, \gamma_{ri}<0}\ \vee_{j,
\beta_{rj}>0} (c_j=0)\right) \right]\, .
\end{equation}
Here, $c_i=0$ in $F_J$ may be read as $i \in J$. Following the previous
section, we use here the notations $\gamma_{ri}$, $\beta_{ri}$ and $\beta_{ri}$
for the reversible reactions and reserve $\nu_l$, $\alpha_{li}^{\nu}$ and
$\beta_{li}^{\nu}$ for the irreversible reactions. The set of $\gamma_r$ in
this formula is the set of the stoichiometric vectors of the reversible
reactions.

The required faces $F_J$ may be constructed in an iterative procedure. First of all, let
us introduce an operation that transforms a set of indexes $S\subset \{1,2,\ldots,n\}$ in
a family of sets, $\mathfrak{S}(S)=\{S'_1,\ldots , S'_{l}\}$.  Let us take formula
(\ref{LogicLimitFace}) and find the set where it is valid for all $i\in S$. This set is
described by the following formula:
\begin{equation}\label{Definmathfrak{S}(S)}
\wedge_{i\in S} \left[(c_i=0) \wedge  \left( \wedge_{r, \gamma_{ri}>0}\
\vee_{j, \alpha_{rj}>0} (c_j=0)\right)  \wedge  \left( \wedge_{r,
\gamma_{ri}<0}\ \vee_{j, \beta_{rj}>0} (c_j=0)\right)\right]\, .
\end{equation}
Let us produce a disjunctive form of this formula and minimize it by
cancelations as it is described in Sec.~\ref{Sec:StSt}. The result is
\begin{equation}
\vee_{j=1,\ldots , k} \left(\wedge_{i\in S_j'} (c_i=0) \right) \, .
\end{equation}
Because of cancelations, the sets $S'_j$ do not include one another. They give the
result, $\mathfrak{S}(S)=\{S'_1,\ldots , S'_{l}\}$. Each $S'_j \in \mathfrak{S}(S)$ is a
superset of $S$, $S' \supseteq S$.

Let us extend the operation $\mathfrak{S}$ on the sets of sets $\mathbf{S}=\{S_1, \ldots,
S_p\}$ with the property: $S_i \not\subset S_j$ for $i \neq j$. Let us apply
$\mathfrak{S}$ to all $S_i$ and take the union of the results:
$\mathfrak{S}_0(\mathbf{S})=\cup_i \mathfrak{S}(S_i)$. Some sets from this  family may
include other sets from it. Let us organize  cancelations: if $S',S'' \in
\mathfrak{S}_0(\mathbf{S})$ and $S' \subset S''$ then retain the smallest set, $S'$, and
delete the largest one. We do the cancelations until it is possible. Let us call the
final result $\mathfrak{S}(\mathbf{S})$. It does not depend on the order of these
operations.

Let us start from any family $\mathbf{S}$ and iterate the operation $\mathfrak{S}$. Then,
after finite number of iterations, the sequence $\mathfrak{S}^d(\mathbf{S})$ stabilizes:
$\mathfrak{S}^d(\mathbf{S})=\mathfrak{S}^{d+1}(\mathbf{S})=\ldots$ because for any set
$S$ all sets from $\mathfrak{S}(S)$ include $S$.

The problems of propositional logic that arise in this and the previous section seem very
similar to elementary logical puzzles \cite{ClarkePuzzles2003}. In the solution we just
use the logical {\em distribution laws} (distribution of conjunction over disjunction and
distribution of disjunction over conjunction), commutativity of disjunction and
conjunction, and elementary cancelation rules like $(A \wedge A) \Leftrightarrow A$, $(A
\vee A) \Leftrightarrow A$, $[A \wedge (A \vee B)]\Leftrightarrow A$, and $[A \vee (A
\wedge B)]\Leftrightarrow A$.

Now, we are in position to describe the construction of all $F_J$ that have the
$\omega$-limit points on their relative interior and are the maximal faces with
this property.
\begin{enumerate}
\item{Let us follow Sec.~\ref{Sec:StSt} and construct the minimized disjunctive form
    (\ref{DisForm1}) for the description of the steady states of the irreversible
    reactions.}
\item{Let us calculate the families of sets $\mathfrak{S}^d(\{S_j\})$ for the family
    of sets $\{S_j\}$ from (\ref{DisForm1}) and $d=1,2,\ldots$, until stabilization.}
\item{Let $\mathfrak{S}^d(\{S_j\})=\mathfrak{S}^{d+1}(\{S_j\}) = \{J_1,J_2, \ldots
    J_p\}$. Then the family of the faces $F_{J_i}$ ($i=1,2, \ldots , p$) gives the
    answer: the $\omega$-limit points are situated in $ri F_{J_i}$ and for each $i$
    there are $\omega$-limit points in $ri F_{J_i}$.}
\end{enumerate}

\subsection{Simple examples}

In this Sec., we present two simple and formal examples of the calculations
described in the previous sections.

1.~$A_1+A_2 \rightleftharpoons A_3+A_4$, $\gamma=(-1,-1,1,1,0)$; $A_1+A_2 \to
A_5$, $\nu=(-1,-1,0,0,1)$. The extended principle of detailed balance holds:
the convex hull of the stoichiometric vectors of the irreversible reactions
consists of one vector $\gamma_2$ and it is linearly independent of $\gamma_1$.
The input  vector $\alpha$ for the irreversible reaction $A_1+A_2 \to A_5$ is
$(-1,-1,0,0,0)$. The set $J=J_l$ from the conjunction form (\ref{COnjFormStSt})
is defined by the non-zero coordinates of this $\alpha^{\nu}$: $J=\{1,2\}$. The
conjunction form in this simple case (one irreversible reaction) loses its
first conjunction operation and is just $(c_1=0)\vee (c_2=0)$. It is, at the
same time, the minimized disjunction form (\ref{DisForm1}) and does not require
additional transformations. This formula describes the steady states of the
irreversible reaction in the positive orthant $\mathbb{R}^n_+$. For this
disjunction form, The family of sets $\mathbf{S}=\{S_j\}$ consists of two sets,
$S_{1}=\{1\}$ and $S_{2}=\{2\}$.

Let us calculate $\mathfrak{S}({S_{1,2}})$. For both cases, $i=1,2$ there are
no reversible reactions with $\gamma_{ri}=0$. Therefore, one expression in
round parentheses vanishes in (\ref{Definmathfrak{S}(S)}). For $S=\{1\}$ this
formula gives
$$(c_1=0)\wedge ((c_3=0)\vee(c_4=0))$$
and for $S=\{2\}$ it gives
$$(c_2=0)\wedge ((c_3=0)\vee(c_4=0))\, .$$

The elementary transformations give the disjunctive forms:
\begin{equation*}
[(c_1=0)\wedge((c_3=0)\vee(c_4=0))]  \Leftrightarrow [((c_1=0)\wedge
(c_3=0))\vee ((c_1=0)\wedge (c_4=0))] \, ,
\end{equation*}
\begin{equation*}
[(c_2=0)\wedge ((c_3=0)\vee(c_4=0))] \Leftrightarrow [((c_2=0)\wedge
(c_3=0))\vee ((c_2=0)\wedge (c_4=0))]\, .
\end{equation*}
 Therefore, $\mathfrak{S}(S_1)=\{\{1,3\}, \{1,4\}\}$, $\mathfrak{S}(S_2)=\{\{2,3\},
\{2,4\}\}$ and
$$\mathfrak{S}(\{S_{1},S_2\})=\{\{1,3\}, \{1,4\},\{2,3\}, \{2,4\}\}\, .$$
No cancelations are needed. The iterations of $\mathfrak{S}$ do not produce new
sets from $\{\{1,3\}, \{1,4\},\{2,3\}, \{2,4\}\}$. Indeed, if $c_1=c_3=0$, or
$c_1=c_4=0$, or $c_2=c_3=0$, or $c_2=c_4=0$ then all the reaction rates are
zero. More formally, for example for $\mathfrak{S}(\{1,3\})$ formula
(\ref{Definmathfrak{S}(S)}) gives
\begin{equation*}
[(c_1=0)\wedge((c_3=0)\vee(c_4=0))] \wedge
[(c_3=0)\wedge((c_1=0)\vee(c_2=0))]\, .
\end{equation*}
This formula is equivalent to
$(c_1=0)\wedge (c_3=0)$. Therefore, $\mathfrak{S}(\{1,3\})=\{1,3\}$. The same
result is true for  $\{1,4\}$, $\{2,3\}$, and $\{2,4\}$.

All the $\omega$-limit points (steady states) belong to the faces
$F_{\{1,3\}}=\{c\, |, c_1=c_3=0\}$, $F_{\{1,4\}}=\{c\, |, c_1=c_4=0\}$,
$F_{\{2,3\}}=\{c\, |, c_2=c_3=0\}$, or $F_{\{2,4\}}=\{c\, |, c_2=c_4=0\}$. The
position of the $\omega$-limit point for a solution $N(t)$ depends on the
initial state. More specifically, this system of reactions has three
independent linear conservation laws: $b_1=N_1+N_2+N_3+N_4+2N_5$, $b_2=N_1-N_2$
and $b_3=N_3-N_4$. For given values of these $b_{1,2,3}$ vector $N$ belongs to
the $2D$ plane in $\mathbb{R}^5$. The intersection of this plane with the
selected faces depends on the signs of $b_{2,3}$:
\begin{itemize}
\item{If $b_2<0$, $b_3<0$ then it intersects $F_{\{1,3\}}$ only, at one point
    $N=(0,-b_2,0,-b_3, b_1+b_2+b_3)$ ($N_5$ should be non-negative, $b_1+b_2+b_3\geq
    0$) .}
\item{If $b_2=0$, $b_3<0$ then it intersects both $F_{\{1,3\}}$ and $F_{\{2,3\}}$ at
    one point $N=(0,0, 0, -b_3, b_1+b_3)$ ($N_5$ should be non-negative, $b_1+b_3\geq
    0$).}
\item{If $b_2<0$, $b_3=0$ then it intersects both $F_{\{1,3\}}$ and $F_{\{1,4\}}$ at
    one point $N=(0,-b_2, 0, 0, b_1+b_2)$ ($N_5$ should be non-negative, $b_1+b_2\geq
    0$).}
\item{If $b_2>0$, $b_3<0$ then it intersects $F_{\{2,3\}}$ only, at one point
    $N=(b_2,0, 0, -b_3, b_1+b_2+b_3)$ ($N_5$ should be non-negative, $b_1+b_2+b_3\geq
    0$).}
\item{If $b_2>0$, $b_3=0$ then it intersects $F_{\{2,3\}}$ and $F_{\{2,4\}}$  at
the point $N=(b_2,0, 0, 0, b_1+b_2)$ ($N_5$ is non-negative because
$b_1+b_2+b_3\geq 0$).}
\item{If $b_2<0$, $b_3>0$ then it intersects $F_{\{1,4\}}$ only, at one point
    $N=(0,-b_2, b_3, 0, b_1+b_2+b_3)$ ($N_5$ should be non-negative, $b_1+b_2+b_3\geq
    0$).}
\item{If $b_2=0$, $b_3>0$ then it intersects $F_{\{1,4\}}$ and $F_{\{2,4\}}$  at one
    point $N=(0,0, b_3, 0, b_1+b_3)$ ($N_5$ is non-negative because $b_1+b_3\geq
    0$).}
\item{If $b_2>0$, $b_3>0$ then it intersects $F_{\{2,4\}}$ only, at one point
    $N=(b_2,0, b_3, 0, b_1+b_2+b_3)$ ($N_5$ is non-negative because $b_1+b_2+b_3\geq
    0$).}
\end{itemize}
As we can see, the system has exactly one $\omega$-limit point for any
admissible combination of the values of the conservation laws. These points are
the listed points of intersection.

For the second simple example, let us change the direction of the irreversible reaction.

2.~$A_1+A_2 \rightleftharpoons A_3+A_4$, $\gamma_1=(-1,-1,1,1,0)$, $A_5 \to
A_1+A_2 $,  $\nu=(1,1,0,0,-1)$. The extended principle of detailed balance
holds. The steady-states of the irreversible reactions is given by one
equation, $c_5=0$. Formula (\ref{Definmathfrak{S}(S)}) gives  for
$\mathfrak{S}(\{5\})$ just $(c_5=0)$. The face $F_{\{5\}}$ includes
$\omega$-limit points in $riF_{\{5\}}$. Dynamics on this face is defined by the
fully reversible reaction system and tends to the equilibrium of the reaction
$A_1+A_2 \rightleftharpoons A_3+A_4$ under the given conservation laws. On this
face, there exist the border equilibria, where $c_1=c_3=0$, or   $c_1=c_4=0$,
or $c_2=c_3=0$, or $c_2=c_4=0$ but they are not attracting the positive
solutions.

\section{Example: H$_2$+O$_2$ system \label{Sec:Hydro}}

For the case study, we selected the H$_2$+O$_2$ system. This is one of the main model
systems of gas kinetics. The hydrogen burning gives us an example of the medium
complexity with 8 components ($A_1=$H$_2$, $A_2=$O$_2$, $A_3=$OH, $A_4=$H$_2$O, $A_5=$H,
$A_6=$O, $A_7=$HO$_2$, and $A_8=$H$_2$O$_2$) and 2 atomic balances (H and O). For the
example, we selected the reaction mechanism from \cite{Vlachos1996}. The literature about
hydrogen burning mechanisms is huge. For recent discussion we refer to
\cite{LiDryer2004,Saxena2006}. We do not aim to compare the different schemes of this
reaction but use this reaction mechanism as an example and a benchmark.

\begin{table}[t]
\centering \caption{H$_2$ burning mechanism \cite{Vlachos1996} \label{TableVlachos}}
{
\begin{tabular}{|l|l|c|}
\hline {No} & {Reaction} & { Stoichiometric vector}
\\ \hline
1   &H$_2$ + O$_2$ $\rightleftharpoons$ 2OH &(-1,-1,2,0,0,0,0,0) \\
2   &H$_2$ + OH $\rightleftharpoons$ H$_2$O + H &(-1,0,-1,1,1,0,0,0)\\
3   &OH + O $\rightleftharpoons$ O$_2$ + H &(0,1,-1,0,1,-1,0,0)  \\
4   &H$_2$ + O $\rightleftharpoons$ OH + H &(-1,0,1,0,1,-1,0,0)\\
5   &O$_2$ + H +M $\rightleftharpoons$ HO$_2$ +M &(0,-1,0,0,-1,0,1,0) \\
6   &OH + HO$_2$  $\rightleftharpoons$ O$_2$ + H$_2$O &(0,1,-1,1,0,0,-1,0)\\
7   &H + HO$_2$ $\rightleftharpoons$ 2OH &(0,0,2,0,-1,0,-1,0)\\
8   &O + HO$_2$ $\rightleftharpoons$ O$_2$ + OH &(0,1,1,0,0,-1,-1,0) \\
9   &2OH $\rightleftharpoons$ H$_2$O + O &(0,0,-2,1,0,1,0,0)\\
10  &2H + M $\rightleftharpoons$ H$_2$ + M &(1,0,0,0,-2,0,0,0)\\
11  &2H + H$_2$ $\rightleftharpoons$ H$_2$ + H$_2$ &(1,0,0,0,-2,0,0,0)\\
12  &2H  + H$_2$O $\rightleftharpoons$ H$_2$ + H$_2$O &(1,0,0,0,-2,0,0,0)\\
13  &OH + H + M $\rightleftharpoons$ H$_2$O + M &(0,0,-1,1,-1,0,0,0)\\
14  &H + O + M $\rightleftharpoons$ OH + M &(0,0,1,0,-1,-1,0,0)\\
15  &2O  + M $\rightleftharpoons$ O$_2$ + M &(0,1,0,0,0,-2,0,0)\\
16  &H + HO$_2$ $\rightleftharpoons$ H$_2$ + O$_2$ &(1,1,0,0,-1,0,-1,0)\\
17  &2HO$_2$ $\rightleftharpoons$ O$_2$ + H$_2$O$_2$ &(0,1,0,0,0,0,-2,1)\\
18  &H$_2$O$_2$  + M $\rightleftharpoons$ 2OH + M &(0,0,2,0,0,0,0,-1)\\
19  &H + H$_2$O$_2$ $\rightleftharpoons$ H$_2$ + HO$_2$ &(1,0,0,0,-1,0,1,-1)\\
20  &OH + H$_2$O$_2$ $\rightleftharpoons$ H$_2$O + HO$_2$ &(0,0,-1,1,0,0,1,-1)\\
\hline
\end{tabular} }
\end{table}

A special symbol ``M" is used for the ``third body". It may be any molecule.
The third body provides the energy balance. Efficiency of different molecules
in this process is different, therefore, the ``concentration" of the third body
is a weighted sum of the concentrations of the components with positive
weights. The third body does not affect the equilibrium constants and does not
change the zeros of the forward and reverse reaction rates but modifies the
non-zero values of reaction rates. Therefore, for our analysis we can omit
these terms. The elementary reactions 10, 11 and 12 are glued in one,
2H$\rightleftharpoons$H$_2$, after cancelation of the third bodies,  and we
analyze the mechanism of 18 reaction.

Under various conditions, some of the reactions are (almost) irreversible and some of
them should be considered as reversible. For example, let us consider the H$_2$+O$_2$
system  at or near the atmospheric pressure and in the temperature interval 800--1200K.
The reactions 1, 2,  4, 18, 19, and 20 are supposed to be reversible (on the base of the
reaction rate constants presented in \cite{Vlachos1996}). The first question is: if these
reactions are reversible then which reactions may be irreversible?

Due to the general criterion, the convex hull of the stoichiometric vectors of
the irreversible reactions has empty intersection with the linear span of the
stoichiometric vectors of the reversible reactions. Therefore, if the
stoichiometric vector of a reaction belongs to the linear span of the
stoichiometric vectors of the reversible reactions, then this reaction is
reversible. Simple linear algebra gives that
$$\gamma_{3,5,9} \in {\rm span} \{\gamma_1, \gamma_2, \gamma_4, \gamma_{18}, \gamma_{19},
\gamma_{20}\}\, .$$ In particular, $\gamma_3=-\gamma_1+\gamma_4$,
$\gamma_5=\gamma_1-\gamma_{18}+\gamma_{19}$,  $\gamma_9=\gamma_2-\gamma_4$. So,
the list of the reversible reactions should include the reactions 1, 2, 3, 4,
5, 9, 18, 19, and 20. The reactions 6, 7, 8, 10, 11, 12, 13, 14, 15, and 17 may
be irreversible. Formally, there are $2^{8}=256$ possible combinations of the
directions of these 8 reactions (the reactions 10, 11 and 12 have the same
stoichiometric vector and, in this sense, should be considered as one
reaction). The general criterion and simple linear algebra give that there are
only two admissible combinations of the directions of irreversible reactions:
either for all of them $k^-_r=0$ or for all of them $k^+_r=0$. Here,  the
forward and reverse reactions and the notations $k^{\pm}_r$ are selected
according to the Table~\ref{TableVlachos}. We can immediately notice that the
inverse direction of all reactions is very far from the reality under the given
conditions, for example, it includes the irreversible dissociation H$_2\to2$H.

Let us demonstrate in detail, how the general criterion produces this reduction from the
256 possible combinations of directions of irreversible reactions to just 2 admissible
combinations. We assume that the initial set of reactions is spit in two: reversible
reactions with numbers $r\in J_0$ and irreversible reactions with $r\in J_1$,
rank$\{\gamma_1,\gamma_2,\ldots , \gamma_{\ell}\}=d$, rank$\{\gamma_r \, | \, r\in
J_0\}=d_0$. The rank of all vectors $\gamma_r$, $d$, must exceed the rank of the
stoichiometric vectors of the reversible reactions, $d>d_0$, because if $d=d_0$ then all
the reactions must be reversible and the problem becomes trivial.

According to \cite{GorbYabCES2012}, we have to perform the following operations with the
set of stoichiometric vectors $\gamma_r$ (Fig.~\ref{Fig1ad}):
\begin{enumerate}
\item{Eliminate several coordinates from all $\gamma_r$ using linear conservation
    laws. This is transfer to the internal coordinates in span$\{\gamma_r\, | \, r=1,
    \ldots , {\ell}\}$;}
\item{Eliminate coordinates from all $\gamma_r$ ($r\in J_1$) using the stoichiometric
    vectors of the reversible reactions and the  {\em Gauss--Jordan elimination}
    procedure. This is the map to the quotient space span$\{\gamma_j\, | \, j=1,
    \ldots , {\ell} \}/\mbox{span}\{\gamma_j\, | \, j\in J_0\}$. Me denote the result
    as $\overline{\gamma}_r$;}
\item{Use the linear programming technique and analyze for which combinations of the
    signs, the convex hull conv$\{\pm \overline{\gamma}_r\, |\, r\in J_1\}$ does not
    include 0.}
\end{enumerate}

\begin{table*}[t]
\centering \caption{Elimination of coordinates of stoichiometric vectors for H$_2$
burning mechanism. The reversible reactions are collected in the upper part of the Table.
The reaction in the lower part of the table are irreversible. The group of equivalent
reactions 10, 11, 12 is presented by one of them. In the second column, the first two
coordinates (which correspond to H$_2$ and O$_2$) are excluded using the atomic balance.
In the following columns the results of the coordinates elimination are presented. For
each step, the pivot for elimination is underlined and highlighted  in bold in the
previous column. The eliminated coordinates at each step are named at the top of each
column.
Their zero values are omitted. \label{TableElimination}} {
\begin{tabular}{|l|c|c|c|c|c|c|}
\hline No &\hspace{-1mm} H$_2$, O$_2$ \hspace{-1mm} &\hspace{-1mm}  OH
\hspace{-1mm}&\hspace{-1mm} H$_2$O$_2$  \hspace{-1mm}&\hspace{-1mm}
H$_2$O \hspace{-1mm}&\hspace{-1mm} H \hspace{-1mm}&\hspace{-1mm} O \hspace{-1mm}
\\ \hline
1   & ($\underline{\mathbf{ 2}}$,0,0,0,0,0) & (0,0,0,0,0)&(0,0,0,0)&(0,0,0)&(0,0)&(0) \\
2   &(-1,1,1,0,0,0) & (1,1,0,0,0)&($\underline{\mathbf{1}}$,1,0,0)&(0,0,0)&(0,0)&(0)\\
3   &(-1,0,1,-1,0,0) &(0,1,-1,0,0)&(0,1,-1,0)  & ($\underline{\mathbf{1}}$,-1,0)&(0,0)&(0)\\
4   &(1,0,1,-1,0,0) &(0,1,-1,0,0)&(0,1,-1,0)&(1,-1,0)&(0,0)&(0)\\
5   &(0,0,-1,0,1,0) & (0,-1,0,1,0)&(0,-1,0,1)&(-1,0,1)&({\bf -}$\underline{\mathbf{1}}$,1)&(0)\\
9   &(-2,1,0,1,0,0) & (1,0,1,0,0)&(1,0,1,0)&(-1,1,0)&(0,0)&(0)\\
18  &(2,0,0,0,0,-1) &(0,0,0,0,{\bf -}$\underline{\mathbf{1}}$)&(0,0,0,0)&(0,0,0)&(0,0)&(0)\\
19  &(0,0,-1,0,1,-1)&(0,-1,0,1,-1)&(0,-1,0,1)&(-1,0,1)&(-1,1)&(0)\\
20  &(-1,1,0,0,1,-1)&(1,0,0,1,-1)&(1,0,0,1)&(-1,0,1)&(-1,1)&(0)\\
\hline \hline
6   &(-1,1,0,0,-1,0)&(1,0,0,-1,0)&(1,0,0,-1)&(-1,0,-1)&(-1,-1)&(-2)\\
7   &(2,0,-1,0,-1,0)&(0,-1,0,-1,0)&(0,-1,0,-1)&(-1,0,-1)&(-1,-1)&(-2)\\
8   &(1,0,0,-1,-1,0)&(0,0,-1,-1,0)&(0,0,-1,-1)&(0,-1,-1)&(-1,-1)&(-2) \\
10  &(0,0,-2,0,0,0)&(0,-2,0,0,0)&(0,-2,0,0)&(-2,0,0)&(-2,0)&(-2)\\
13  &(-1,1,-1,0,0,0)&(1,-1,0,0,0)&(1,-1,0,0)&(-2,0,0)&(-2,0)&(-2)\\
14  &(1,0,-1,-1,0,0)&(0,-1,-1,0,0)&(0,-1,-1,0)&(-1,-1,0)&(-2,0)&(-2)\\
15  &(0,0,0,-2,0,0)&(0,0,-2,0,0)&(0,0,-2,0)&(0,-2,0)&(-2,0)&(-2)\\
16  &(0,0,-1,0,-1,0)&(0,-1,0,-1,0)&(0,-1,0,-1)&(-1,0,-1)&(-1,-1)&(-2)\\
17  &(0,0,0,0,-2,1)&(0,0,0,-2,1)&(0,0,0,-2)&(0,0,-2)&(0,-2)&(-2)\\
\hline
\end{tabular}                    }
\end{table*}

In the Table~\ref{TableElimination} we present the results of the step-by-step
elimination. First, the atomic balances give us for every possible stoichiometric vector
$\eta=(\eta_1, \ldots ,\eta_8)$ two identities:
\begin{enumerate}
\item{$2\eta_1+\eta_3+2\eta_4+\eta_5+\eta_7+2\eta_8 = 0$ or $\eta_1=
    -\frac{1}{2}(\eta_3+2\eta_4+\eta_5+\eta_7+2\eta_8)$;}
\item{$2 \eta_2+\eta_3+\eta_4+\eta_6+2\eta_7+2\eta_8=0$ or $ \eta_2=
    -\frac{1}{2}(\eta_3+\eta_4+\eta_6+2\eta_7+2\eta_8)$.}
\end{enumerate}
Let us recall that the order of the coordinates $(\eta_1, \ldots ,\eta_8)$ corresponds to
the following order of the components, (H$_2$, O$_2$, OH, H$_2$O, H, O, HO$_2$,
H$_2$O$_2$). Due to these identities, a stoichiometric vector $\eta$ for this mixture is
completely defined by six coordinates $(\eta_3, \ldots ,\eta_8)$. In the second column of
the Table~\ref{TableElimination} these 6D vectors are given for all the reactions from
the H$_2$ burning mechanism (the Table~\ref{TableVlachos}).

In five columns No. 3-7, the results of the coordinate eliminations are presented (and
the zero-valued eliminated coordinates are omitted). Each elimination step may be
represented as a projection:
$$x\mapsto x - x_i \frac{1}{\eta_i} \eta \, ,$$
where $\eta_i$ is a {\em pivot} (highlighted in bold in the column preceding
the result of elimination), and $\eta$ is the vector that includes the pivot
(as the $i$th coordinate). The projection operator is applied to every vector
of the previous column. At the end (the last column), all the stoichiometric
vectors of the reversible reaction are transformed into zero, and the
stoichiometric vectors of the irreversible reactions with the given direction
(from the left to the right) are transformed into the same vector $(-2)$. If we
restore all the zeros, then the corresponding 6D vector is $(0,0,0,0,-2,0)$. We
have to use the atomic balances to return to the 8D vectors. The coordinate
$x_7$ corresponds to HO$_2$, $x_1$ corresponds to H$_2$, and $x_2$ corresponds
to O$_2$, hence, $2x_1-2=0$ and $2x_2-4=0$. The restored 8D vector is
$(1,2,0,0,0,0,-2,0)$.

A convex combination of several copies of one vector cannot give zero.
Therefore, the structural condition of the extended principle of detailed
balance holds. It holds also for the inverse direction of all the irreversible
reactions. All other distributions of directions can produce zero in the convex
hull and are inadmissible. So, we have the following list of irreversible
reactions that satisfies the extended principle of detailed balance for given
reversible reactions. (We will not discuss the second list of reverse
irreversible reactions because it has not much sense for given conditions.)

\begin{tabular}{l l}
6   &OH + HO$_2$  $\to$ O$_2$ + H$_2$O \\
7   &H + HO$_2$ $\to$ 2OH \\
8   &O + HO$_2$ $\to$ O$_2$ + OH \\
10  &2H $\to$ H$_2$ \\
13  &OH + H $\to$ H$_2$O \\
14  &H + O $\to$ OH \\
15  &2O  $\to$ O$_2$ \\
16  &H + HO$_2$ $\to$ H$_2$ + O$_2$ \\
17  &2HO$_2$ $\to$ O$_2$ + H$_2$O$_2$.
\end{tabular}

We assume that all the reaction rate constants for the selected directions are strictly
positive. The rate of all these reaction vanish if and only if concentration of H, O and
HO$_2$ are equal to zero, $c_{5,6,7}=0$. Indeed, $c_{5}=0$ if and only if $w_{10}=0$,
$c_{6}=0$ if and only if $w_{15}=0$, $a_{7}=0$ if and only if $w_{17}=0$. On the other
hand, all other reaction rates from this list are zeros if $c_{5,6,7}=0$.

Let us reproduce this reasoning using formulas from Sec.~\ref{Sec:StSt}. For the $l$th
irreversible reaction, $J_l$ is the set of indexes $i$ for which $\alpha_{li}\neq 0$. Let
us keep for the irreversible reactions their numbers (6, 7, 8, 10, 13, 14, 15, 16, 17).
For them, $J_6=\{3,7\}$, $J_7=\{5,7\}$, $J_8=\{6,7\}$, $J_{10}=\{5\}$, $J_{13}=\{3,5\}$,
$J_{14}=\{5,6\}$, $J_{15}=\{6\}$, $J_{16}=\{5,7\}$, $J_{17}=\{7\}$.

Formula (\ref{DisForm1}) gives for the steady states of the irreversible reactions:
\begin{equation*}
\begin{split}
&((c_3=0)\vee(c_7=0))\wedge((c_5=0)\vee(c_7=0)) \wedge ((c_6=0)\vee(c_7=0)) \wedge (c_5=0) \\
  \wedge &((c_3=0)\vee (c_5=0)) \wedge ((c_5=0)\vee(c_6=0)) \wedge (c_6=0)   \wedge ((c_5=0)\vee(c_7=0))\wedge(c_7=0).
\end{split}
\end{equation*}
It is equivalent to $$(c_5=0)\wedge(c_6=0)\wedge(c_7=0)\, .$$ Of course, the result is
the same, the face $F_{\{5,6,7\}}$ ($c_{5,6,7}=0$, $c_i \geq 0$) is the set of the steady
states of all irreversible reaction.

Let  us look now on the list of reversible reactions:

\begin{tabular}{l l}
1   &H$_2$ + O$_2$ $\rightleftharpoons$ 2OH \\
2   &H$_2$ + OH $\rightleftharpoons$ H$_2$O + H\\
3   &OH + O $\rightleftharpoons$ O$_2$ + H \\
4   &H$_2$ + O $\rightleftharpoons$ OH + H \\
5   &O$_2$ + H $\rightleftharpoons$ HO$_2$ \\
9   &2OH $\rightleftharpoons$ H$_2$O + O \\
18  &H$_2$O$_2$  $\rightleftharpoons$ 2OH \\
19  &H + H$_2$O$_2$ $\rightleftharpoons$ H$_2$ + HO$_2$ \\
20  &OH + H$_2$O$_2$ $\rightleftharpoons$ H$_2$O + HO$_2$ \\
\end{tabular}

If the concentration OH ($c_3$) is positive then the component O is produced in the
reaction 9. If the concentrations of H$_2$ ($c_1$) and OH ($c_3$) both are positive then
the component H is produced in reaction 2. If the concentrations of H$_2$O$_2$ ($c_8$)
and OH ($c_3$) both are positive then the component HO$_2$ is produced in reaction 2. Due
to the reversible reaction 18 any of two components H$_2$O$_2$ and OH produces the other
component. Moreover, the first reaction produces OH from H$_2$ + O$_2$. This production
stops if and only if either concentration of H$_2$ is zero ($c_1=0$) or concentration of
O$_2$ is zero ($c_2=0$).

This means that the set of zeros of the irreversible reactions, $c_{5,6,7}=0$ ($c\geq
0$), is {\em not invariant} with respect to the kinetics of the reversible reactions.
This means that from an initial conditions on this set the kinetic trajectory will leave
it unless, in addition, $c_3=c_8=0$ and either $c_1=0$ or $c_2=0$.

The reactions of all irreversible reactions should tend to zero due to
Proposition~\ref{LyapunovIrrev}. Therefore, the kinetic trajectory should
approach the union of two planes, $c_{1,3,5,6,7,8}=0$ and $c_{2,3,5,6,7,8}=0$
(under condition $c\geq 0$). These planes are two-dimensional and the position
of the state there is completely defined by the atomic balances.

If the concentration vector belongs to the first plane, then all the atoms are collected
in O$_2$ and H$_2$O. It is possible if and only if $b_{\rm O} \geq \frac{1}{2}b_{\rm H}$.
In this case, $c_4=\frac{1}{2}b_{\rm H}$ and $c_2=\frac{1}{2}(b_{\rm O}-\frac{1}{2}b_{\rm
H})$.

If the concentration vector belongs to the second plane, then all the atoms are collected
in H$_2$ and H$_2$O. It is possible if and only if $b_{\rm O} \leq \frac{1}{2}b_{\rm H}$.
In this case, $c_4=b_{\rm O}$ and $c_1=\frac{1}{2}(b_{\rm H}-2b_{\rm O})$.

Let us reproduce this reasoning formally using the general formalism of
Sec.~\ref{Sec:InvStSt}. Formula~\ref{Definmathfrak{S}(S)} gives for
$\mathfrak{S}(\{5,6,7\})$
\begin{equation}\label{mathfrak{S}5,6,7)}
\begin{split}
(c_5=0)\wedge &\left( \wedge_{r, \gamma_{r5}>0}\vee_{j, \alpha_{rj}>0}
(c_j=0)\right)  \wedge \left( \wedge_{r,\gamma_{r5}<0}\ \vee_{j, \beta_{rj}>0} (c_j=0)\right) \\
\wedge  (c_6=0)\wedge &\left( \wedge_{r, \gamma_{r6}>0}\vee_{j, \alpha_{rj}>0}
(c_j=0)\right)
\wedge \left( \wedge_{r,\gamma_{r6}<0}\ \vee_{j, \beta_{rj}>0} (c_j=0)\right) \\
\wedge  (c_7=0)\wedge &\left( \wedge_{r, \gamma_{r7}>0}\vee_{j, \alpha_{rj}>0}
(c_j=0)\right)  \wedge \left( \wedge_{r, \gamma_{r7}<0}\ \vee_{j, \beta_{rj}>0}
(c_j=0)\right)\, .
\end{split}
\end{equation}
Vectors $\gamma_r$ that participate in this formula are the stoichiometric vectors of
reversible reactions ($r=1,2,3,4,5,9,18,19,20$). From the Table~\ref{TableVlachos} we
find that
 $\gamma_{r5}>0$ for $r= 2, 3, 4$,  $\gamma_{r5}<0$ for $r= 5, 19$,
 $\gamma_{r6}>0$ for  $r= 9$, $\gamma_{r6}<0$ for $r= 3,4$,
 $\gamma_{r7}>0$ for $r= 5, 19,20$, and $\gamma_{r7}\not<0$ for all $r$.
Formula (\ref{mathfrak{S}5,6,7)}) transforms into
\begin{equation*}
\begin{split}
&(c_5=0)\wedge ((c_1=0) \!\vee\! (c_3=0))\wedge((c_3=0) \!\vee\!(c_6=0))
 \wedge ((c_1=0) \!\vee\!(c_6=0)) \wedge (c_7=0) \\
  \wedge & ((c_1=0) \!\vee\!(c_7=0))
\wedge (c_6=0)\wedge(c_3=0)\wedge((c_2=0) \!\vee\!(c_5=0))
\wedge ((c_3=0) \!\vee\! (c_5=0)) \\
 \wedge & (c_7=0)\wedge((c_2=0) \!\vee\!(c_5=0))
\wedge  ((c_5=0) \!\vee\!(c_8=0))\wedge((c_3=0) \!\vee\!(c_8=0))\, .
\end{split}
\end{equation*}
After simple transformations it becomes
\begin{equation}\label{mathfrak5,6,7simple}
(c_3=0)\wedge(c_5=0)\wedge(c_6=0) \wedge(c_7=0)\, .
\end{equation}
Therefore, $\mathfrak{S}(\{5,6,7\})=\{3,5,6,7\}$. To iterate, we have to compute
$\mathfrak{S}(\{3,5,6,7\})$. For this calculation, we have to add one more line to
formula (\ref{mathfrak{S}5,6,7)}), namely,
\begin{equation*}
\wedge (c_3=0)\wedge \left( \wedge_{r, \gamma_{r3}>0}\vee_{j, \alpha_{rj}>0}
(c_j=0)\right) \wedge \left( \wedge_{r, \gamma_{r3}<0} \vee_{j, \beta_{rj}>0}
(c_j=0)\right)\, .
\end{equation*}
Let us take into account that
 $\gamma_{r3}>0$ for $r= 1, 4, 18$ and  $\gamma_{r3}<0$ for $r= 2,3,9,20$, and rewrite
 this formula in the more explicit form
\begin{equation*}
\begin{split}
&(c_3=0)\wedge((c_1=0)\vee(c_2=0))  \wedge ((c_1=0)\vee(c_6=0))\wedge(c_8=0) \\
 \wedge &((c_4=0)\vee(c_5=0))\wedge((c_2=0)\vee(c_5=0)) \wedge ((c_4=0)\vee(c_6=0))\wedge(c_7=0)\, .
\end{split}
\end{equation*}
Let us take the conjunction of this formula with (\ref{mathfrak{S}5,6,7)}) taken in the
simplified equivalent form (\ref{mathfrak5,6,7simple}) and transform the result to the
disjunctive form. We get
\begin{equation}
\begin{split}
[&(c_3=0)\wedge(c_5=0)\wedge(c_6=0) \wedge(c_7=0)\wedge(c_8=0)\wedge(c_1=0)]\\
\vee [&(c_3=0)\wedge(c_5=0)\wedge(c_6=0)
\wedge(c_7=0)\wedge(c_8=0)\wedge(c_2=0))]\, .
\end{split}
\end{equation}
This means that
$\mathfrak{S}^2(\{5,6,7\})=\mathfrak{S}(\{3,5,6,7\})=\{\{1,3,5,6,7,8\},\{2,3,5,6,7,8\}\}$.
The further calculations show that the next iteration does not change the result.
Therefore, all the $\omega$-limit points belong to two faces, $F_{\{1,3,5,6,7,8\}}$ and
$F_{\{2,3,5,6,7,8\}}$. The result is the same as for the previous discussion. The
detailed formalization becomes crucial for more complex systems and for software
development.

Let us find the vector of exponents $\delta=(\delta_i)$ ($i=1,\ldots, 8$) from the
Table~\ref{TableElimination}. After all the eliminations, the corresponding linear
functional $\hat{\delta}$ is just a value of the 7th coordinate: $\hat{\delta}(x)=x_7$.
Its values are negative ($-2$) for all irreversible reactions and zero for all reversible
reactions (see the last column of the Table~\ref{TableElimination}). The conditions
$(\delta,\gamma)=0$ for the reversible reactions and $(\delta,\gamma)<0$ for all
irreversible reactions do not define the unique vector: if $\delta$ satisfies these
conditions then its linear combination with the vectors of atomic balances also satisfy
them. Such a combination is a vector
\begin{equation}\label{GaugeInv}
\lambda \delta+\lambda_{\rm H}(2,0,1,2,1,0,1,2)+\lambda_{\rm
O}(0,2,1,1,0,1,2,2)\,
\end{equation}
under condition $\lambda>0$. This transformation of $\delta$ does not change
the signs of $\hat{\delta}$ on the stoichiometric vectors because of atomic
balances.

In our case  the only coordinate remains not eliminated, $x_7$ (the bottom part
of the last column of the Table~\ref{TableElimination}). If, for some reaction
mechanism and  selected sets  of reversible and irreversible reaction, there
remain several ($q$) coordinates, then it is necessary to find $q$
corresponding functionals $\hat{\delta}$ and the space of possible vectors of
exponents is ($q+j$)-dimensional. Here, $j$ is the number of the independent
linear conservation laws for the whole system, $j=n-{\rm rank}\{\gamma_r\}$,
$n$ is the number of the components, $\{\gamma_r\}$ includes all the
stoichiometric vectors for reversible and irreversible reactions.

To find $\delta$, we apply the elimination procedures from the
Table~\ref{TableElimination} to an arbitrary vector $y=(y_i)$ ($i=1, \ldots,
8$):
\begin{equation}
\begin{split}
&(y_1,y_2,y_3,y_4,y_5,y_6,y_7,y_8) \mapsto (y_1,y_2,0,y_4,y_5,y_6,y_7,y_8) \\
\mapsto &(y_1,y_2,0,y_4,y_5,y_6,y_7,0)\mapsto (y_1,y_2,0,0,y_5-y_4,y_6,y_7,0)\\
\mapsto &(y_1,y_2,0,0,0,y_6+y_5-y_4,y_7,0) \mapsto
(y_1,y_2,0,0,0,0,y_7+y_6+y_5-y_4,0) \, .
\end{split}
\end{equation}
This sequence of transformations gives us the linear functional
$$\hat{\delta}(y)=y_7+y_6+y_5-y_4\, .$$ The corresponding vector of exponents $(0,0,0,-1,1,1,1,0)$ should be
corrected because its coordinates cannot be negative. Let us apply
(\ref{GaugeInv}) with $\lambda=2$ (for convenience).  The coordinates of this
combination are non-negative if and only if $\lambda_{\rm H} \geq 0$,
$\lambda_{\rm O}\geq 0$ and $2\lambda_{\rm H}+\lambda_{\rm O}-2\geq 0$. The
solutions of these linear inequality on the $(\lambda_{\rm H},\lambda_{\rm O})$
plane is a convex combination of the extreme points (corners) $(1,0)$ and
$(0,2)$ plus any non-negative 2D vector: $(\lambda_{\rm H},\lambda_{\rm
O})=\varsigma (1,0) +(1-\varsigma)(0,2)+(\vartheta_1,\vartheta_2)$,
$\vartheta_{1,2}\geq 0$ and $1\geq \varsigma \geq 0$. The corresponding vectors
of exponents are
\begin{equation*}
(0,0,0,-2,2,2,2,0)+(\varsigma+\vartheta_1) (2,0,1,2,1,0,1,2)
+(1-\varsigma+\vartheta_2)(0,4,2,2,0,2,4,4)\, .
\end{equation*}
At least one of the exponents should be zero. There are only three
possibilities, $\delta_{1}$, $\delta_{2}$ or $\delta_{4}$. For all other $i$,
$\delta_i>0$ if $\vartheta_{1,2}\geq 0$ and $1\geq \varsigma \geq 0$.

To provide any necessary atomic balance in the limit $\varepsilon \to 0$ it is
necessary that two of $\delta_i$ are zeros. If $b_{\rm O} \leq
\frac{1}{2}b_{\rm H}$, then  $\delta_1=\delta_4=0$. This means that
$\vartheta_{1,2}= 0$, $\varsigma = 0$ and $\delta=(0,4,2,0,2,4,6,4)$. It is
convenient to divide this $\delta$ by 2 and write
$$\delta=(0,2,2,0,2,2,3,2)\, .$$
For these exponents, the equilibrium concentrations tend to 0 with the small parameter
$\varepsilon\to 0$ ($\varepsilon > 0$) as
\begin{equation}
\begin{split}
&c^{\rm eq}_{{\rm H}_2}=c^{\rm eq}_1=const, \,
 c^{\rm eq}_{{\rm O}_2}=c^{\rm eq}_2\sim\varepsilon^2, \,
 c^{\rm eq}_{{\rm OH}}=c^{\rm eq}_3\sim\varepsilon^2, \,
 c^{\rm eq}_{{\rm H}_2{\rm O}}=c^{\rm eq}_4 = const,\, \\
 &c^{\rm eq}_{{\rm H}}=c^{\rm eq}_5\sim\varepsilon^2, \,
 c^{\rm eq}_{{\rm O}}=c^{\rm eq}_6\sim\varepsilon^2, \,
 c^{\rm eq}_{{\rm H}{\rm O}_2}=c^{\rm eq}_7\sim \varepsilon^3, \,
 c^{\rm eq}_{{\rm H}_2{\rm O}_2}=c^{\rm eq}_6\sim \varepsilon^2 \, .
\end{split}
\end{equation}

If $b_{\rm O} \geq \frac{1}{2}b_{\rm H}$, then $\delta_2=\delta_4=0$. This means that
$\vartheta_{1,2}= 0$, $\varsigma = 1$ and
$$\delta=(2,0,1,0,3,2,3,2)\, .$$
For these exponents, the equilibrium concentrations tend to 0 with the small parameter
$\varepsilon\to 0$ ($\varepsilon > 0$) as
\begin{equation}
\begin{split}
&c^{\rm eq}_{{\rm H}_2}=c^{\rm eq}_1\sim\varepsilon^2,\,
 c^{\rm eq}_{{\rm O}_2}=c^{\rm eq}_2=const, \,
 c^{\rm eq}_{{\rm OH}}=c^{\rm eq}_3\sim\varepsilon, \,
 c^{\rm eq}_{{\rm H}_2{\rm O}}=c^{\rm eq}_4 = const, \, \\
 &c^{\rm eq}_{{\rm H}}=c^{\rm eq}_5\sim\varepsilon^3, \,
 c^{\rm eq}_{{\rm O}}=c^{\rm eq}_6\sim\varepsilon^2, \,
 c^{\rm eq}_{{\rm H}{\rm O}_2}=c^{\rm eq}_7\sim \varepsilon^3, \,
 c^{\rm eq}_{{\rm H}_2{\rm O}_2}=c^{\rm eq}_6\sim \varepsilon^2 \, .
\end{split}
\end{equation}

The linear combination $\sum_i \delta_i N_i$ decreases in time due to kinetic equations.
This is true for any vector of exponents presented by a linear combination
(\ref{GaugeInv}) ($\lambda\neq 0$) of the initial vector $(0,0,0,-1,1,1,1,0)$ with the
vectors of the atomic balances. At the same time, any of these combinations give an
additional linear conservation law for the system of reversible reactions.

Below are several versions of this function:
\begin{itemize}
\item{The initial version, $\hat{\delta}$, obtained from the
    Table~\ref{TableElimination} is $(\delta,N)=-N_{{\rm H}_2 {\rm 0}}+N_{\rm
    H}+N_{\rm O}+N_{\rm H} {{\rm O}_2}$;}
\item{Vector of exponents, calibrated by adding of the atomic balances
    (\ref{GaugeInv}) to meet the atomic balance conditions for $b_{\rm O} \leq
    \frac{1}{2}b_{\rm H}$ in the limit $\varepsilon \to 0$  is $(\delta,N)=2N_{{\rm
    O}_2} +2N_{\rm OH}+2N_{\rm H}+2N_{\rm O}+3N_{{\rm H} {\rm O}_2}+2N_{{\rm H}_2
    {\rm O}_2}$;}
\item{Vector of exponents, calibrated to meet the atomic balance conditions for
    $b_{\rm O} \geq \frac{1}{2}b_{\rm H}$ is, $(\delta,N)=2N_{{\rm H}_2} +N_{\rm
    OH}+3N_{\rm H}+2N_{\rm O}+3N_{{\rm H} {\rm O}_2}+2N_{{\rm H}_2 {\rm O}_2}$.}
\end{itemize}
All these forms differs by the combinations of the atomic balances
(\ref{GaugeInv}) and are, in this sense, equivalent.

\section{Conclusion \label{Sec:Conclus}}

The general principle of detailed balance was formulated in 1925 as follows
\cite{Lewis1925}: ``Corresponding to every individual process there is a reverse process,
and in a state of equilibrium the average rate of every process is equal to the average
rate of its reverse process."  Rigorously speaking, the chemical reactions have to be
considered as reversible ones, and every step of the complex reaction consists of two
reactions, forward and reverse (backward) one. However, in reality the rates of some
forward or reverse reactions may be negligible. Typically, the complex combustion
reactions, in particular, reactions of hydrocarbon oxidation or hydrogen combustion,
include both reversible and irreversible steps.  It is a case in catalytic reactions as
well. In particular, many enzyme reactions are ``partially irreversible". Although many
catalytic reactions are globally irreversible, they always include some reversible steps,
in particular steps of adsorption of gases. Many enzyme reactions are also ``partially
irreversible".

This work aims to solve the {\em problem of the partially irreversible limit in chemical
thermodynamics} when some reactions become irreversible whereas some other reactions
remain reversible. The main results in this direction are
\begin{enumerate}
\item{Description of the multiscale limit of a system reversible reactions when some
    of equilibrium concentrations tend to zero (Sec.~\ref{Sec:Shifted}).}
\item{Extended principle of detailed balance for the systems with some irreversible
    reactions (Theorem~\ref{Theorem:ExPrDetBalDegenerat}).}
\item{The linear functional $G_{\delta}$ that decreases in time on solutions of the
    kinetic equations under the extended detailed balance conditions
    (Proposition~\ref{LyapunovIrrev} and Eq. (\ref{LinearDiss})).}
\item{The entropy production (or free energy dissipation) formulas for the reversible
    part of the reaction mechanism under the extended detailed balance conditions
    (Eqs. (\ref{DissClass}), (\ref{DissTanh})). }
\item{Description of the faces of the positive orthant which include the
    $\omega$-limit points in their relative interior and, therefore, description of
    limiting behavior in time (Theorem~\ref{LocationOmegaLimit}).}
\end{enumerate}

Did we solve the main problem and create the thermodynamic of the systems with some
irreversible reaction? The answer is: we solved this problem partially. We described the
limit behavior but we did not find the global Lyapunov function that captures relaxation
of both reversible and irreversible parts of the system. The good candidate is a linear
combination of the relevant classical thermodynamic potential and $G_{\delta}$ but we did
not find the coefficients. In that sense, the problem of the limit thermodynamics remains
open.

Nevertheless, one problem is solved ultimately and completely: {\em How to throw away
some reverse reactions without violation of thermodynamics and microscopic
reversibility?} The answer is: the convex hull of the stoichiometric vectors of the
irreversible reactions should not intersect with the linear span of the stoichiometric
vectors of the reversible reactions and the reaction rate constants of the remained
reversible reactions should satisfy the Wegscheider identities (\ref{WegscheiderLambda}).

The solution of this theoretical problem is important for the modeling of the chemical
reaction networks. This is because some of reactions are practically irreversible.
Removal of some reverse reaction from the reaction mechanism cannot be done independently
of the whole structure of the reaction network; the whole reaction mechanism should be
used in the decision making.

If the irreversible reactions are introduced correctly then we also know that the closed
system with this reaction mechanism goes to an equilibrium state. At this equilibrium,
all the reaction rates are zero: the irreversible reaction rates vanish and the rates of
the reversible reactions satisfy the principle of detailed balance. The limit equilibria
are situated on the faces of the positive orthant of concentrations and these faces are
described in the paper.

The oscillatory or chaotic attractors are impossible in closed systems which satisfy the
extended principle of detailed balance. This general statement can be considered as a
simple consequence of thermodynamics. It can be easily proved if the thermodynamic
Lyapunov functions (potentials) are given. However, the thermodynamic potentials have no
limits for the systems with some irreversible reactions and we do not know a priori any
general theorem that prohibits bifurcations at the zero values of some reaction rate
constants. In this paper we proved, in particular, that the emergency of nontrivial
attractors in systems with some irreversible reactions is impossible if they are the
limits of the reversible systems which satisfy the principle of detailed balance. In this
sense, the thermodynamic behavior is proven for the systems with some irreversible
reactions under the extended detailed balance conditions. Nevertheless, the general
problem of the thermodynamic potentials in this limit remains open.

\end{document}